\documentclass[prl,superscriptaddress,amsmath,amssymb,twocolumn,nofootinbib]{revtex4-2}

\usepackage[T1]{fontenc}
\usepackage{bm,amsmath}
\usepackage[usenames,dvipsnames]{xcolor}
\usepackage[normalem]{ulem}

\usepackage[colorlinks,citecolor=NavyBlue,linkcolor=NavyBlue,urlcolor=NavyBlue,]{hyperref}

\usepackage{amsthm}
\usepackage{enumerate}
\usepackage{orcidlink}

\def\sgn{\operatorname{sgn}}

\newtheorem{theorem}{Theorem}
\newtheorem{corollary}{Corollary}
\newtheorem{lemma}{Lemma}
\newtheorem{proposition}{Proposition}

\newcommand{\titleinfo}{
Williamson majorization theory of fermionic non-Gaussianity}

\newcommand{\Tr}{\text{Tr}}   
\newcommand{\tr}{\operatorname{tr}}
\newcommand{\rank}{\operatorname{rank}}
\newcommand{\Ran}{\operatorname{Image}}
\newcommand{\supp}{\operatorname{supp}}
\newcommand{\id}{\mathbb I}
\newcommand{\Gam}{\Gamma}
\newcommand{\cJ}{\mathfrak J}
\newcommand{\cH}{\mathcal H}
\newcommand{\Faf}{\mathcal F}
\newcommand{\rvec}{\bm r}
\newcommand{\ketbra}[2]{|#1\rangle\!\langle#2|}
\newcommand{\nuG}{\nu_{\rm G}}
\newcommand{\Mocc}{M_{\rm occu}}

\begin{document}
\title{\titleinfo}

\author{ 
Xhek Turkeshi~\orcidlink{0000-0003-1093-3771}}
\email{turkeshi@thp.uni-koeln.de}
\affiliation{Institut f\"ur Theoretische Physik, Universit\"at zu K\"oln, Z\"ulpicher Strasse 77, 50937 K\"oln, Germany}

\author{Piotr Sierant~\orcidlink{0000-0001-9219-7274}}
\affiliation{Barcelona Supercomputing Center, Barcelona 08034, Spain}

\author{ 
Poetri Sonya Tarabunga~\orcidlink{0000-0001-8079-9040}}
\email{poetri.tarabunga@tum.de}
\affiliation{Technical University of Munich, TUM School of Natural Sciences,
Physics Department, 85748 Garching, Germany}
\affiliation{Munich Center for Quantum Science and Technology (MCQST),
Schellingstr. 4, 80799 M\"unchen, Germany}
    
\begin{abstract}
Pure-state entanglement rests on a single algebraic backbone: majorization of the Schmidt spectrum governs state conversion  under local operations and classical communication, and constrains entanglement monotones. 
Here we establish a corresponding majorization law for fermionic non-Gaussianity, the resource that elevates free fermions to universal quantum computation. 
Under any fermionic Gaussian protocol with pure state outcomes, the Williamson spectrum of a pure state's Majorana covariance matrix is weakly majorized by its ensemble average.
This spectral law mirrors that of entanglement theory. It turns computable non-Gaussianity quantifiers such as fermionic antiflatness and occupation entropies into strong monotones for fermionic non-Gaussianity, and delivers necessary conditions and converse bounds on state conversion under Gaussian protocols.
When fermion parity is conserved, no catalyst can remove a majorization obstruction---unless it carries parity coherence---and asymptotic interconversion is irreversible already for pure states. 
All relevant quantities are accessible from two-point Majorana correlators,  turning the theory developed here into experimentally observable properties of quantum matter, testable on present-day quantum devices.
\end{abstract}

\maketitle
\textit{Introduction.---}
Majorization of the Schmidt spectrum provides a unifying framework for manipulating bipartite pure-state entanglement~\cite{ChitambarGour2019,AmicoFazioOsterlohVedral2008,HorodeckiEtAl2009,NielsenVidal2001}.
Nielsen's theorem shows that it completely determines whether one pure state can be converted deterministically into another by local operations and classical communication (LOCC)~\cite{Nielsen1999}. 
When deterministic conversion is impossible, the same spectral structure determines the best achievable success probability as shown by Vidal~\cite{Vidal1999}, and it also underlies more general transformations into ensembles~\cite{JonathanPlenio1999a} and catalytic transformations, in which an auxiliary entangled state enables otherwise impossible conversions without being consumed~\cite{JonathanPlenio1999b,Klimesh2007,Turgut2007,KondraDattaStreltsov2021}.
In the many-copy limit, pure-state entanglement manipulation becomes reversible, with the entropy of entanglement setting the conversion rate~\cite{BennettEtAl1996,PopescuRohrlich1997,DonaldHorodeckiRudolph2002}. 
Thus, much of pure-state entanglement manipulation can ultimately be understood in terms of the ordering of Schmidt spectra. The same majorization order recurs across resource theories: pure-state coherence under incoherent operations~\cite{DuBaiGuo2015,StreltsovAdessoPlenio2017}, athermality under thermal operations~\cite{HorodeckiOppenheim2013,BrandaoEtAl2015}, and nonuniformity under noisy operations~\cite{GourEtAl2015}.

This work establishes the analogous ordering principle for \emph{fermionic non-Gaussianity}. Here, fermionic Gaussian states---the states of non-interacting fermions, generated by matchgate circuits~\cite{Valiant2002, TerhalDiVincenzo2002, Knill2001, JozsaMiyake2008} and completely described by their Majorana covariance matrix~\cite{Bravyi2005,DiVincenzoTerhal2005,PeschelEisler2009,SuraceTagliacozzo2022}---constitute the free sector, and the non-Gaussianity injected by interactions is the resource. The distinction between Gaussian and genuinely interacting states underlies mean-field descriptions of fermionic matter~\cite{FetterWalecka03,Negele18quantum,AltlandSimons10,deGennes99,Schrieffer18}, 
controls the cost of classical simulation~\cite{DiasKoenig2024, ReardonSmithOszmaniecKorzekwa2024}, governs the universality and quantum advantage of fermionic
computation~\cite{bravyi2002fermionic,HebenstreitJozsaKrausStrelchuk2019,
PhysRevA.102.052604,Oszmaniec22}, and sets the complexity of learning and
certifying fermionic states~\cite{PhysRevLett.120.190501,
Bittel2025paclearningoffree,MeleHerasymenko2025}. Computable and measurable quantifiers of this resource, including the fermionic antiflatness~\cite{SierantStornatiTurkeshi2026} and the occupation entropies~\cite{TarabungaEtAl2026}, are built from the Williamson spectrum of the covariance matrix~\cite{GigenaDiTullioRossignoli2020,SierantStornatiTurkeshi2026,TarabungaEtAl2026,HaugTurkeshiSierant2026,Pachos2022quantifying,PhysRevResearch.6.023176,debertolis2025naturalsuperorbitalsrepresentationmanybody,lyu2024fermionicgaussiantestingnongaussian,coffman2025measuringnongaussianmagicfermions}, with strong monotonicity under Gaussian protocols established only in two specific cases~\cite{LeoneBittel2026,TarabungaEtAl2026}. 
We show that, analogously to the Schmidt spectrum in entanglement theory, the Williamson spectrum obeys a fundamental majorization law: under any Gaussian protocol applied to a pure state, the input spectrum is weakly majorized by the ensemble-averaged spectrum of the outputs (Theorem~\ref{thm:main}). 

This single spectral law has several immediate consequences. 
It turns every convex Williamson deficit into a strong monotone, including the full hierarchies of fermionic antiflatness and occupation entropies, and reduces the characterization of spectral monotones to a condition on their generating function. 
It also yields Nielsen- and Vidal-type conversion bounds that are optimal within the class of spectral monotones, and shows that Gaussian nullity remains monotone even under postselection. 
When fermion parity is conserved, the structure becomes even sharper: Williamson spectra concatenate, making the spectral measures additive, and no parity-definite catalyst can enlarge the majorization order. Parity coherence, a genuinely fermionic resource, can restore catalysis, but the resulting gain is universally bounded by one maximally non-Gaussian mode. Meanwhile asymptotic interconversion can be strongly irreversible already among pure states.

\textit{Setting.---}
We consider $N$ fermionic modes, described by $2N$ Majorana operators $\gamma_a=\gamma_a^\dagger$ with $\{\gamma_a,\gamma_b\}=2\delta_{ab}$, and characterize a state $\rho$ by its covariance matrix
\begin{equation}
  (\Gam_\rho)_{ab}:=-\frac{i}{2}\Tr\bigl(\rho[\gamma_a,\gamma_b]\bigr),
  \label{eq:covariance}
\end{equation}
a real antisymmetric matrix that an orthogonal change of Majorana basis brings to the canonical form $\bigoplus_j \left(\begin{smallmatrix}0&r_j\\-r_j&0\end{smallmatrix}\right)$ with $1\ge r_1\ge\cdots\ge r_N\ge0$~\cite{Williamson1936,Bhatia1997}. The Williamson values $r_j(\Gam_\rho)$ are invariant under Gaussian unitaries, which act as special orthogonal rotations of the Majorana basis; positivity of the state gives $r_j\le1$, and a pure state is fermionic Gaussian if and only if every $r_j$ equals one~\cite{Bravyi2005,SuraceTagliacozzo2022}. The deviation of the spectrum $\rvec(\Gam_\rho)=(r_1,\ldots,r_N)$ from the Gaussian spectrum $\rvec_\mathrm{G}=(1,\ldots,1)$ therefore quantifies pure-state non-Gaussianity. The free operations are Gaussian protocols~\cite{TarabungaEtAl2026,tarabunga2026fermionic,LeoneBittel2026}: Gaussian unitaries, pure Gaussian ancillas, occupation-number measurements, mode discards, and classical feedforward, mapping an input $\rho$ to an outcome ensemble $\{(p_s,\rho_s)\}_s$. 
Throughout the main text we restrict to protocols whose outcome ensembles consist of pure states, denoted $\{(p_s,\phi_s)\}_s$; mixed outcomes are discussed in the SM~\cite{SuppMat}.
A nonnegative functional $M$ is a pure-state \emph{strong monotone} if $M(\psi)\ge\sum_sp_sM(\phi_s)$ for every Gaussian protocol mapping a pure input $\psi$ to such an ensemble~\cite{ChitambarGour2019,Nielsen1999}.

\textit{Williamson measures.---}
Every non-Gaussianity measure we consider is a sum over the spectrum of a scalar function
applied to each Williamson value, measured from unity. For
$f\colon [0,1]\to\mathbb R$ we call
\begin{equation}
  \Phi_f(\rho):=N f(1)-\sum_{j=1}^{N}f\!\left(r_j(\Gam_\rho)\right)
  =\sum_{j=1}^{N}\left[f(1)-f\!\left(r_j(\Gam_\rho)\right)\right]
  \label{eq:fw-definition}
\end{equation}
the \emph{$f$-Williamson measure} of $\rho$. We refer to $\Phi_f$, and more generally to any nonnegative functional of the spectrum that vanishes on the Gaussian ceiling $\rvec_\mathrm{G}$, as a \textit{spectral deficit}.
Being a function of the
spectrum alone, $\Phi_f$ is automatically invariant under Gaussian
unitaries, and it is manifestly computable from two-point correlators.
Two choices of $f$ recover the quantifiers already in use: $f(t)=t^{2k}$
gives the fermionic antiflatness $\Faf_k$ of Ref.~\cite{SierantStornatiTurkeshi2026}, and
$f(t)=-h_\alpha\!\left(\frac{1+t}{2}\right)$, with ${h_\alpha(p)=\tfrac{p^\alpha + (1-p)^\alpha-1}{1-\alpha}}$ the binary Tsallis entropy, gives the
occupation-number entropies $\Mocc^{[\alpha]}$ of
Ref.~\cite{TarabungaEtAl2026}.

\textit{Williamson majorization.---}
We now establish the key result of this work, a majorization relation
obeyed by the Williamson spectrum. We write $\widetilde\rvec$ for the Williamson spectrum padded by unit entries to a common length, accounting for Gaussian ancillas.
For a covariance matrix $\Gam$ we write
${S_\ell(\Gam):=\sum_{j\le\ell} \widetilde r_j(\Gam)}$ for the sum of the largest $\ell$ (padded)
Williamson values, which we denote as partial sums. For nonnegative decreasing vectors $x,y$ we write
$x\prec_w y$ when $\sum_{j\le\ell}x_j\le\sum_{j\le\ell}y_j$ for all
$\ell$~\footnote{This is weak majorization, which, unlike ordinary majorization, does
not require equality of the total
sums~\cite{Bhatia1997, MarshallOlkinArnold2011}}.

\begin{theorem}[Williamson majorization]
\label{thm:main}
Let $\psi$ be a pure $N$-mode state and let $\{(p_s,\phi_s)\}_{s}$ be
the ensemble produced by any Gaussian protocol.
Then
\begin{equation}
  S_\ell(\Gam_{{\psi}})\ \le\ \sum_{s}p_s\,S_{\ell}(\Gam_{{\phi_s}})
  \qquad \forall\,\ell=1,\ldots,N,
  \label{eq:partialsum-main}
\end{equation}
or equivalently, in vector form,
\begin{equation}
  \widetilde\rvec(\Gam_{{\psi}})\ \prec_w\
  \Bigl(\textstyle\sum_s p_s \widetilde r_1(\Gam_{{\phi_s}}),\ldots,
  \sum_s p_s \widetilde r_{N}(\Gam_{{\phi_s}})\Bigr).
  \label{eq:vector-majorization}
\end{equation}
\end{theorem}

\begin{proof}[Proof sketch]
Gaussian unitaries preserve the spectrum, ancillas append unit values, and adaptive protocols compose by induction over the protocol tree, with spectra of different lengths compared after unit padding. Everything therefore condenses into a single question: \textit{What happens when one mode is measured in its occupation basis?} For that step, we prove, for an arbitrary state $\rho$,
\begin{equation}
  S_\ell(\Gam_\rho)\ \le\ 1+\sum_sp_s\,S_{\ell-1}(\Gam_{\rho_s}'),
  \label{eq:one-mode}
\end{equation}
with $\Gam'_{\rho_s}$ the branch covariance of the unmeasured modes and $S_0:=0$---the measured mode carries at most one unit of spectral weight, the Williamson value of its recorded occupation state. 
The proof of~\eqref{eq:one-mode} rests on three ingredients. 
First, the partial sums
are variational, $S_\ell(\Gam)=\max_J\frac12\tr(J^T\Gam)$ over rank-$2\ell$
partial complex structures $J$~\cite{Bhatia1997}. Thus, fixing a maximizer turns $S_\ell$ into the expectation of the quadratic observable
$Q_J=-\frac{i}{2}\sum_{ab}J_{ab}\gamma_a\gamma_b=\sum_{j} g_j$, a sum
of commuting reflections $g_j$ in the appropriate basis. Second, the Majorana plane
of the measured mode meets at most two complex planes of $J$, confining the
nontrivial part of the problem to at most three modes, on which the reduced
state is mixed (the reason the underlying lemma must be, and is, proven
for mixed inputs). Third, a Gaussian-code inequality closes the residual two-plane case: $Q_J\le2P_{++}$, the occupation postselections of the code projector $P_{++}$ are rank-deficient Gaussian operators, and the product normal form converts this deficiency into branch reflections $h_s$ with $Q_J\le\id+\sum_s\ketbra{s}{s}\otimes h_s$ as an operator inequality. Taking expectations yields Eq.~\eqref{eq:one-mode}. 
Complete proofs are provided in the Supplemental Material (SM)~\cite{SuppMat}.
\end{proof}

Notably, Theorem~\ref{thm:main} requires no fermionic parity assumption. It also partially holds for arbitrary mixed states, as discussed in the SM~\cite{SuppMat}.

\textit{Spectral monotones.---}
Because Theorem~\ref{thm:main} constrains the entire Williamson
spectrum, it controls every functional built from that spectrum at once, and
the measures $\Phi_f$ of Eq.~\eqref{eq:fw-definition} that qualify as
monotones can be characterized completely. The characterization is a
condition on the scalar function $f$ alone (proof in SM~\cite{SuppMat}).

\begin{theorem}[Characterization of $f$-Williamson monotones]
\label{thm:admissible}
Let $f:[0,1]\to\mathbb R$ be nondecreasing, so that $\Phi_f\ge0$. Then
\begin{enumerate}[(i)]\itemsep1pt
\item $\Phi_f$ is a strong pure-state non-Gaussianity monotone
\emph{if and only if} $f$ is convex.
\item $\Phi_f$ is additive on any tensor product in which at least one
factor has definite parity.
\item $\Phi_f$ is pure-state faithful \emph{if and only if} $f(r)<f(1)$ for every
$r<1$.
\item $\Phi_f$ is asymptotically continuous whenever $f$ is continuous.
\end{enumerate}
We call $f$ \emph{admissible} when it meets all of these conditions, i.e.\
when it is continuous and convex on $[0,1]$ with $f(r)<f(1)$ for every
$r<1$. An admissible $\Phi_f$ is then simultaneously faithful,
additive, strongly monotone, and asymptotically continuous.
\end{theorem}

For entanglement, these four requirements single out one functional, the entropy of entanglement~\cite{PopescuRohrlich1997,DonaldHorodeckiRudolph2002}; here they leave an infinite-dimensional convex cone of mutually nonproportional admissible $f$, so no distinguished measure of fermionic non-Gaussianity emerges.

A notable, discontinuous member of the broader family of $f$-Williamson monotones is the Gaussian nullity
$\nuG=\#\{j:r_j<1\}$, which counts the Williamson values less than unity and thus records the number of modes which carry non-Gaussian correlations~\cite{MeleHerasymenko2025}. It corresponds to the step generator $f=\mathbf 1_{\{r=1\}}$,
which is nondecreasing and convex, so Theorem~\ref{thm:admissible}(i)
implies that it is a strong monotone. For this member, however, a stronger statement holds: the nullity is monotone under postselection.
Analogous postselection monotonicity holds for the stabilizer nullity under Clifford operations~\cite{Beverland2020} and the symplectic rank under bosonic Gaussian protocols~\cite{mele2026symplectic}; the fermionic case was still open and resolved here.

\begin{proposition}[Nullity monotonicity under postselection]
\label{prop:nullity}
Let $\psi$ be pure and let $\psi_s$ be the postselected state of any nonzero
branch of a Gaussian protocol. Then
\begin{equation}
  \nuG(\psi_s)\ \le\ \nuG(\psi) .
  \label{eq:nullity-branchwise}
\end{equation}
\end{proposition}

\begin{proof}
Write $\ell=N-\nuG(\psi)$. The partial sum $S_\ell$  satisfies $S_\ell\le\ell$, with equality exactly when those $\ell$
values all equal unity, that is when $\ell\le N-\nuG$. The choice of $\ell$
gives $S_\ell(\Gam_\psi)=\ell$, so Theorem~\ref{thm:main} forces
$\sum_sp_sS_\ell(\Gam_{\psi_s})\ge\ell$; since every term obeys
$S_\ell(\Gam_{\psi_s})\le\ell$, each nonzero branch attains
$S_\ell(\Gam_{\psi_s})=\ell$, that is $\ell\le N-\nuG(\psi_s)$. 
\end{proof}

An alternative proof, by counting linear Majorana annihilators, is given in the SM~\cite{SuppMat}.

\textit{State conversion.---}
Theorem~\ref{thm:main}
is reminiscent of the condition for converting a pure bipartite state
into an ensemble of pure states by LOCC, namely that the Schmidt coefficients of
the initial state be majorized by the probability-weighted average of those
of the final states~\cite{JonathanPlenio1999a,NielsenVidal2001}. It immediately yields
the following majorization condition for deterministic Gaussian
transformations. 

\begin{corollary}[Deterministic conversion]
\label{cor:deterministic}
Let a deterministic Gaussian protocol map the pure state $\psi$ to the pure state $\phi$. Then
\begin{equation}
  \widetilde\rvec(\Gam_\psi)\ \prec_w\ \widetilde\rvec(\Gam_\phi) .
  \label{eq:deterministic-majorization}
\end{equation}
\end{corollary}

Corollary~\ref{cor:deterministic} is the fermionic counterpart of Nielsen's condition for deterministic LOCC transformations~\cite{Nielsen1999}, with the Schmidt spectrum replaced by the Williamson spectrum. 
In contrast to~\cite{Nielsen1999}, however, the majorization is only necessary, not sufficient: satisfying it does not by itself guarantee that a Gaussian conversion exists. The reason is that the Williamson spectrum is not a complete invariant of fermionic states. 
For example,
the GHZ states $|{\rm GHZ}_N\rangle=(|0\rangle^{\otimes N}+|1\rangle^{\otimes N})/\sqrt2$~\cite{Greenberger07} have identically vanishing covariance for every $N\ge3$, so that
\begin{equation}
  \rvec\bigl(\Gam_{{\rm GHZ}_6^{\otimes2}}\bigr)
  =\rvec\bigl(\Gam_{{\rm GHZ}_4^{\otimes3}}\bigr)=(0,\ldots,0).
  \label{eq:ghz-incomplete}
\end{equation}
Thus Eq.~\eqref{eq:deterministic-majorization} holds in both directions, even though ${\rm GHZ}_6^{\otimes2}$ cannot be converted into ${\rm GHZ}_4^{\otimes3}$ by Gaussian protocols~\cite{tarabunga2026fermionic}.

Theorem~\ref{thm:main} also constrains \emph{probabilistic} conversion. To
this end, we define, for $\ell=1,\ldots,N$,
\begin{equation}
  V_\ell(\rho):=\ell-S_\ell(\Gam_\rho)=\sum_{j\le\ell}\bigl(1-r_j(\Gam_\rho)\bigr).
  \label{eq:partialsum-deficit}
\end{equation}
Subtracting Eq.~\eqref{eq:partialsum-main} from $\ell$ shows that each $V_\ell$
is a strong monotone. As in Theorem~\ref{thm:main}, when two
states are compared the deficits are evaluated on the unit-padded spectra,
denoted $\widetilde V_\ell$; padding carries no deficit weight but aligns
the indices~\cite{SuppMat}. Since the success branch carries the deficit of $\phi$ alone (recorded and ancilla modes contribute one, hence no deficit) while failure branches contribute nonnegatively, this yields the following,
\begin{corollary}[Single-shot conversion]
\label{cor:single-shot}
If a Gaussian protocol maps the pure state $\psi$ into the pure state
$\phi$ with probability $p_{\rm succ}$, then
\begin{equation}
  p_{\rm succ}\ \le\ \min_{\ell\,:\, \widetilde V_\ell(\phi)>0}
  \frac{ \widetilde V_\ell(\psi)}{ \widetilde V_\ell(\phi)}.
  \label{eq:vidal-bound}
\end{equation}
\end{corollary}

Corollary~\ref{cor:single-shot} mirrors Vidal's single-copy formula, where the analogous expression is an equality~\cite{Vidal1999}. Here the expression can only be a bound, since equality would make the spectrum a complete invariant, which Eq.~\eqref{eq:ghz-incomplete} forbids. Every strong monotone
yields a bound of this form, and we now show that Corollary~\ref{cor:single-shot}  nonetheless gives the strongest bound that can
be obtained from the Williamson spectrum. Call $M$ a \emph{spectral
monotone} if it is a nonnegative function of the Williamson spectrum that
vanishes on pure Gaussian states, is unchanged by unit padding, and is
nondecreasing and jointly concave
in the variables $1-r_j$; in particular, every $f$-Williamson measure is of
this type, being a sum of concave scalar functions. 

\begin{proposition}[Optimality among spectral monotones]
\label{prop:optimal}
Let $c(\psi,\phi)$ denote the right-hand side of Eq.~\eqref{eq:vidal-bound},
and assume $c(\psi,\phi)\le1$. Then, for every spectral monotone $M$ and pure non-Gaussian states $\psi,\phi$,
\begin{equation} 
  c(\psi,\phi)\ \le\ \frac{M(\psi)}{M(\phi)} ,
  \label{eq:optimality}
\end{equation}
so that Corollary~\ref{cor:single-shot} provides the optimal bound
obtainable from a spectral monotone of this class.
\end{proposition}

\begin{proof}
Write $u_j=1-r_j$ for the unit-padded spectra, ordered increasingly, so
that $\widetilde V_\ell=\sum_{j\le\ell}u_j$ is the sum of the $\ell$
smallest of them, and abbreviate $c=c(\psi,\phi)$. By definition
$\sum_{j\le\ell}u_j(\psi)\ge c\sum_{j\le\ell}u_j(\phi)$ for every $\ell$,
that is, $u(\psi)$ is weakly supermajorized by $c\,u(\phi)$. A spectral
monotone is symmetric and concave, hence Schur-concave, and is nondecreasing
in each argument, so $M(u(\psi))\ge M(c\,u(\phi))$. Concavity together with
$M(0)=0$ gives $M(c\,u)\ge c\,M(u)$ for $c\le1$, and combining the two yields
$M(\psi)\ge c\,M(\phi)$.
\end{proof}

We note that the $V_\ell$ are not $f$-Williamson measures, and are generally not additive even for parity-preserving systems. They are therefore not useful for multi-copy state conversion, in particular for asymptotic transformations. That regime is instead governed by the $f$-Williamson measures, which we will return to below.

\textit{Catalysis.---}
Despite not being sufficient, Corollary~\ref{cor:deterministic} provides a useful no-go statement: if $\widetilde\rvec(\Gam_\psi)\not\prec_w\widetilde\rvec(\Gam_\phi)$, no Gaussian protocol maps $\psi$ to $\phi$. In entanglement theory such obstructions can be circumvented by \emph{catalysis}: an auxiliary $\omega$ is adjoined and returned unaltered, chosen so that the majorization relation, failing for $\psi$ and $\phi$, holds for $\psi\otimes\omega$ and $\phi\otimes\omega$, which Nielsen's sufficiency converts into an LOCC protocol~\cite{JonathanPlenio1999b}. \textit{Is the same strategy available here?} The following theorem answers in the negative whenever the catalyst has definite fermionic parity.

\begin{theorem}[No catalyzed
majorization with parity-preserving catalysts]
\label{thm:nocatalysis}
Let $\psi,\phi$ be pure states with
$\widetilde\rvec(\Gam_\psi)\not\prec_w\widetilde\rvec(\Gam_\phi)$. Then no
pure state catalyst $\omega$ of definite parity enables the conversion
$\psi\otimes\omega\to\phi\otimes\omega$.
\end{theorem}

\begin{proof}
By Corollary~\ref{cor:deterministic} the catalyzed conversion would require
$\widetilde\rvec(\Gam_{\psi\otimes\omega})\prec_w
\widetilde\rvec(\Gam_{\phi\otimes\omega})$. Since $\omega$ has definite
parity, its odd Majorana expectations vanish, the covariance of each product
is block diagonal, and the Williamson spectra concatenate, $\widetilde\rvec(\Gam_{\psi\otimes\omega})
  =\widetilde\rvec(\Gam_\psi)\sqcup\rvec(\Gam_\omega)$,
and likewise for $\phi$. For equal-length vectors---which unit padding
ensures, contributing equal $f(1)$ terms to both sides---weak majorization
holds precisely when
$\sum_jf(x_j)\le\sum_jf(y_j)$ for every nondecreasing convex $f$, and such a
sum splits over a disjoint union. The contribution of $\rvec(\Gam_\omega)$
is therefore common to the two sides and cancels, so the catalyzed
majorization holds if and only if
$\widetilde\rvec(\Gam_\psi)\prec_w\widetilde\rvec(\Gam_\phi)$ does, contrary
to hypothesis.
\end{proof}

The mechanism is precisely where entanglement and non-Gaussianity theories deviate.  Schmidt coefficients multiply under tensor products, so $\sum_jf(\lambda_j)$ is not additive and the additive entropies characterize the weaker catalytic order known as trumping~\cite{Klimesh2007,Turgut2007}.  Williamson values instead concatenate, so $\sum_jf(r_j)$ is additive and no
catalyst can alter the majorization relation.

Parity coherence, however, reopens the possibility, since concatenation of
the Williamson spectra ceases to hold. The $H$-gadget of
Ref.~\cite{BrodGalvao2012} is a concrete example: the matchgate $G(H,H)$,
which acts as $H$ on both parity sectors, satisfies
\begin{equation}
  G(H,H)\,\bigl(|\chi\rangle\otimes|+\rangle\bigr)
  =\bigl(H|\chi\rangle\bigr)\otimes|+\rangle
  \label{eq:catalysis}
\end{equation}
for every single-mode state $\chi$. The catalyst is returned unaltered while
the first mode is transformed by a Hadamard, which is not Gaussian, so a catalyst
promotes a matchgate to a non-Gaussian gate, deterministically and
reversibly, since $G(H,H)^2=\id$. In particular, the catalytic conversion $|0\rangle\otimes|+\rangle \to |+\rangle\otimes|+\rangle$ is allowed by deterministic Gaussian protocol, despite $|0\rangle \to |+\rangle$ being forbidden by $\widetilde\rvec(\Gam_0)\not\prec_w\widetilde\rvec(\Gam_+)$. However, even without parity constraint, the gain from catalysis is sharply bounded, as shown in the following theorem.

\begin{theorem}[Bounded catalytic gain]
\label{thm:catalytic-bound}
Let $\psi,\phi$ be pure states and let $\omega$ be any pure catalyst. If a Gaussian protocol maps
$\psi\otimes\omega\to\phi\otimes\omega$, then
\begin{equation}
  S_{\ell+1}(\Gam_\psi)\ \le\ S_{\ell}(\Gam_\phi)+1
  \qquad\forall\,\ell ,
  \label{eq:catalytic-bound}
\end{equation}
equivalently
$\widetilde\rvec(\Gam_\psi)\sqcup(0)\prec_w\widetilde\rvec(\Gam_\phi)\sqcup(1)$.
\end{theorem}

\begin{proof}[Proof sketch]
The
covariance matrix of a product has the block form
\begin{equation}
  \Gam_{\rho\otimes\omega}
  =\begin{pmatrix}
      \Gam_\rho & u_\rho v_\omega^T\\
      -v_\omega u_\rho^T & \Gam_\omega
    \end{pmatrix},
  \label{eq:product-cross-block}
\end{equation}
where $u_\rho$ and $v_\omega$ collect odd Majorana expectations.
Let $W=v_\omega^\perp$ in the catalyst Majorana space and let $B_W$ be an
isometry onto $W$. With $R=\id\oplus B_W$, the codimension-one compression
$C_\rho:=R^T\Gam_{\rho\otimes\omega}R$ has vanishing cross block and hence
$C_\rho=\Gam_\rho\oplus\Gam_{\omega,W}$. If $\bm d_\omega$ denotes the
paired singular values of the odd-dimensional matrix $\Gam_{\omega,W}$,
write $z_\rho=\widetilde\rvec(\Gam_\rho)\sqcup\bm d_\omega$. Cauchy interlacing theorem~\cite{Bhatia1997}
for the corresponding Hermitian matrices therefore gives
$\widetilde r_j(\Gam_{\rho\otimes\omega})\ge(z_\rho)_j\ge
\widetilde r_{j+1}(\Gam_{\rho\otimes\omega})$. Combining the two interlacing relations
with the necessary majorization condition for the catalyzed conversion
yields $z_\psi\sqcup(0)\prec_w z_\phi \sqcup(1)$. The common contribution
$\bm d_\omega$ cancels by the convex-function characterization of weak
majorization, giving the claim. See SM~\cite{SuppMat} for the complete
proof.
\end{proof}

In other words, a catalyst is worth at most
one maximally non-Gaussian mode adjoined to the input, regardless of the
catalyst and with no parity assumption on any of the three states. In particular, if
$\widetilde\rvec(\Gam_\psi)\sqcup(0)\not\prec_w
\widetilde\rvec(\Gam_\phi)\sqcup(1)$,
then no catalyst can enable the conversion $\psi\to\phi$. 
Catalytic gain is therefore sharply bounded at the level of the Williamson spectrum.

\textit{Asymptotic Gaussian transformations.---}
We now turn to the asymptotic regime, in which Gaussian protocols map
$\psi^{\otimes n}$ into $\phi^{\otimes m_n}$ \emph{approximately}, with an
error that vanishes as the number of copies grows~\cite{ChitambarGour2019}.
Concretely, the protocol outputs a pure state $\widetilde\phi_n$ with
$\|\widetilde\phi_n-\phi^{\otimes m_n}\|_1\le\epsilon_n\to0$, along a sequence of deterministic Gaussian protocols, and the
optimal rate of the conversion is the supremum over all such sequences,
\begin{equation}
  R(\psi\to\phi):=\limsup_{n\to\infty}\frac{m_n}{n} .
  \label{eq:rate-def}
\end{equation}
Any monotone that is additive, strongly monotone and asymptotically
continuous upper bounds such a rate~\cite{ChitambarGour2019}. These three
properties are established for the $f$-Williamson monotones in
Theorem~\ref{thm:admissible}; the parity restriction in the additivity
statement, needed for $\Phi_f(\psi^{\otimes n})=n\Phi_f(\psi)$, is
automatic under 
parity-conserving dynamics, which we assume below.

\begin{corollary}[Asymptotic conversion rate]
\label{cor:rate}
For pure non-Gaussian states $\psi$ and $\phi$ of definite fermionic parity,
\begin{equation}
  R(\psi\to\phi)\le\inf_{f}\frac{\Phi_f(\psi)}{\Phi_f(\phi)},
  \label{eq:general-rate}
\end{equation}
where the infimum runs over all admissible $f$ with $\Phi_f(\phi)>0$.
\end{corollary}

The optimization in Eq.~\eqref{eq:general-rate} can in fact be reduced to a finite one. For $c\in[0,1)$, define the truncated affine generator $f_c(r):=\max\{r,c\}$, which is admissible, with the corresponding Williamson monotone
\begin{equation}
W_c(\rho):=\Phi_{f_c}(\rho)
=\sum_j\min\bigl\{1-r_j(\Gamma_\rho),\,1-c\bigr\}.
\label{eq:capped-Williamson}
\end{equation}

\begin{proposition}[Finite optimization of the asymptotic Williamson bound] \label{prop:finite_opt}
Let $\psi$ and $\phi$ be pure non-Gaussian states of definite fermionic parity. Then
\begin{equation}
R(\psi\to\phi)
\leq
\min_{c\in\mathcal C_{\psi,\phi}}
\frac{W_c(\psi)}{W_c(\phi)},
\label{eq:finite-Williamson-bound}
\end{equation}
where
\begin{equation}
\mathcal C_{\psi,\phi}
:=
{0}
\cup
\{r_j(\Gamma_\psi):r_j(\Gamma_\psi)<1 \}
\cup
\{r_j(\Gamma_\phi):r_j(\Gamma_\phi)<1\}.
\label{eq:threshold-set}
\end{equation}
\end{proposition}

The proof is given in the SM~\cite{SuppMat}. It is instructive to consider the two endpoints of this optimization. At $c=0$, one has
$W_0(\rho)
=
\sum_j\bigl[1-r_j(\Gamma_\rho)\bigr]$.
On the other hand, if $c<1$ is at least as large as every nonunit Williamson value carried by $\psi$ and $\phi$, then $W_c(\rho)=(1-c)\nu_{\rm G}(\rho)$.
Eq.~\eqref{eq:finite-Williamson-bound} therefore gives
\begin{equation}
R(\psi\to\phi)
\leq
\frac{\nu_{\rm G}(\psi)}
{\nu_{\rm G}(\phi)}.
\label{eq:nullity-rate-bound}
\end{equation}
Although the nullity generator itself is discontinuous, for every fixed pair of states the bound in Eq.~\eqref{eq:nullity-rate-bound} is obtained from a continuous admissible generator $f_c$.

The two endpoints become optimal whenever one of the states has vanishing covariance matrix. First, suppose that $\Gamma_\psi=0$ on $N$ modes. Then
$W_c(\psi)=N(1-c)$ and $W_c(\phi)
\leq
(1-c)\nu_{\rm G}(\phi)$,
and therefore
\begin{equation}
\frac{W_c(\psi)}{W_c(\phi)}
\geq
\frac{N}{\nu_{\rm G}(\phi)}.
\end{equation}
If instead $\Gamma_\phi=0$ on $N$ modes, then $W_c(\phi)=N(1-c)$, while $\frac{\min\{1-r,1-c\}}{1-c}
\geq 1-r$
for every $r,c\in[0,1)$. Consequently,
\begin{equation}
\frac{W_c(\psi)}{W_c(\phi)}
\geq
\frac{W_0(\psi)}{N},
\end{equation}
with equality at $c=0$. We thus obtain the closed-form optimized bounds
\begin{equation}
  \Gam_\psi=0\Rightarrow R\le\frac{N}{\nuG(\phi)},
  \qquad
  \Gam_\phi=0\Rightarrow R\le\frac{W_0(\psi)}{N},
  \label{eq:two-directions}
\end{equation}
with $R:= R(\psi\to\phi)$.

The bounds of Eq.~\eqref{eq:two-directions} have implications for the
reversibility of state conversion in the approximate setting. Let $\omega$
be a pure state on $N$ modes with $\Gam_\omega=0$, and let $\chi$ be any
non-Gaussian pure state. Multiplying the two directions,
\begin{equation}
  R(\omega\to\chi)\,R(\chi\to\omega)\ \le\
  \frac{W_0(\chi)}{\nuG(\chi)}\ \le\ 1 ,
  \label{eq:round-trip}
\end{equation}
the final inequality holding because each nonunit Williamson value
contributes at most one to $W_0(\chi)$ and exactly one to $\nuG(\chi)$.

\begin{theorem}[Irreversibility]
\label{thm:irrev}
The resource theory of fermionic non-Gaussianity is irreversible under
approximate asymptotic Gaussian protocols: there exist pure states
$\omega,\chi$ for which the product
$R(\omega\to\chi)\,R(\chi\to\omega)$ is strictly smaller than one. That
product can moreover be made arbitrarily small.
\end{theorem}

Both claims follow from Eq.~\eqref{eq:round-trip}. The inequality
$W_0(\chi)\le\nuG(\chi)$ is strict whenever $\chi$ carries a Williamson
value in the open interval $(0,1)$, which gives the first statement; and
letting the nonunit values of $\chi$ approach one drives
$W_0(\chi)/\nuG(\chi)$ to zero, which gives the second. The irreversibility
is therefore arbitrarily strong, as the product of the two rates admits no positive lower bound. 
An analogous statement for exact postselected
conversion was obtained in Ref.~\cite{tarabunga2026fermionic}.

The contrast with entanglement is stark. Bipartite pure-state entanglement
is asymptotically \emph{reversible}, the rate of conversion between two pure
states being the ratio of their entanglement
entropies~\cite{BennettEtAl1996}, and irreversibility sets in only in the
mixed-state setting~\cite{VidalCirac2001}.
Theorem~\ref{thm:irrev} shows that fermionic non-Gaussianity is irreversible
already at the pure-state level.

\textit{Conclusions.---}
Williamson majorization organizes the manipulation theory of fermionic
non-Gaussianity around a single spectral order: deterministic-conversion
conditions, single-shot bounds, complete monotone families, catalytic constraints, and
asymptotic rates all descend from Theorem~\ref{thm:main}, in close formal
parallel to pure-state entanglement yet with important differences. 
Concatenation makes the spectral monotones additive and therefore abundant,
so that no distinguished measure of fermionic non-Gaussianity exists; it
forbids catalytic majorization in the parity-superselected theory; and it renders
asymptotic interconversion arbitrarily irreversible already for pure
states. Parity coherence---a genuinely fermionic phenomenon with no
entanglement counterpart---reopens catalysis, but only to a sharply bounded extent: any catalyst can supply at most one maximally non-Gaussian mode.

Our results open several avenues for exploration. States that no spectral monotone distinguishes may still be non-interconvertible, their obstruction certified only by monotones of higher order in the Majorana correlators, such as the bridge degree of Ref.~\cite{tarabunga2026fermionic}: is there a complete set of monotones for pure-state Gaussian convertibility, furnished perhaps by the commutant of the Gaussian group on several copies~\cite{sierant2026theorymatchgatecommutant,braccia2026commutant,lastres2026geometry}? Which conversion rates are achievable, i.e., how far are the converse bounds from optimal, both in the single-shot and asymptotic setting? Does the majorization order survive for bosonic Gaussian resources~\cite{TakagiZhuang2018,AlbarelliGenoniSerafiniFerraro2018,LamiEtAl2018}, where the compactness $r_j\le1$ of the fermionic spectrum has no counterpart? Finally, every quantity in this Letter is built from two-point Majorana correlators and is measurable for pure states with presently available matchgate control~\cite{SierantStornatiTurkeshi2026,HaugTurkeshiSierant2026} (on mixed states, spectral deficits additionally reflect classical mixedness): the conversion bounds derived here are experimentally accessible tools for tracking how interactions generate fermionic magic in many-body dynamics~\cite{TirritoTurkeshiSierant2024,PaviglianitiEtAl2026,trigueros2026unitarydesigns,santra2025quantumresourcesnonabelianlattice,falcao2026fermionicmagicresourcesdisordered,ares2026nongaussianityrandomquantumstates,79vj-nx6r,aditya2025mpembaeffectsquantumcomplexity,ares2026asymmetrylowerboundfermionic,matsuda2026quantumcomputationalresourcesconformal,cqfh-mbwk,zavatti2026quantummagicstronglycorrelated, Swietek26one, bhakuni2026quantumresourcesdisorderfreelocalization,bera2015mbl}.

\begin{acknowledgments}
\textit{Acknowledgments.---}
P.S.T. thanks Ra{\'u}l Morral-Yepes and Marc Langer for useful discussions.
X.T. acknowledges support from DFG Emmy Noether Programme proposal
``Digital Quantum Matter Out-of-Equilibrium'' No. 560726973, DFG under
Germany's Excellence Strategy -- Cluster of Excellence Matter and Light for
Quantum Computing (ML4Q) EXC 2004/2 -- 390534769, and DFG Collaborative
Research Center (CRC) 183 Project No. 277101999 -- project B01.
P.S. acknowledges fellowship within the ``Generaci\'on D'' initiative,
Red.es, Ministerio para la Transformaci\'on Digital y de la Funci\'on
P\'ublica, for talent attraction (C005/24-ED CV1), funded by the European
Union NextGenerationEU funds, through PRTR. P.S.T. acknowledges funding from the European Research Council (ERC) under the European Union (ERC, DynaQuant, No. 101169765).

ChatGPT 5.6 Sol was used in the development of the proof of Theorem~\ref{thm:main}. The resulting proof was independently verified by the authors. Claude Fable 5 was used for assistance with manuscript writing.

\end{acknowledgments}

\bibliography{FAF_strong_monotone_refs}

@misc{bhakuni2026quantumresourcesdisorderfreelocalization,
      title={Quantum Resources in Disorder-Free Localization Dynamics of Gauge Theories}, 
      author={Devendra Singh Bhakuni and Giovanni Cataldi and Jad C. Halimeh and Emanuele Tirrito},
      year={2026},
      eprint={2607.28730},
      archivePrefix={arXiv},
      primaryClass={quant-ph},
      url={https://arxiv.org/abs/2607.28730}, 
}

@misc{coffman2025measuringnongaussianmagicfermions,
      title={Measuring Non-Gaussian Magic in Fermions: Convolution, Entropy, and the Violation of Wick's Theorem and the Matchgate Identity}, 
      author={Luke Coffman and Graeme Smith and Xun Gao},
      year={2025},
      eprint={2501.06179},
      archivePrefix={arXiv},
      primaryClass={quant-ph},
      url={https://arxiv.org/abs/2501.06179}, 
}

@article{Pachos2022quantifying,
  doi = {10.22331/q-2022-10-13-840},
  url = {https://doi.org/10.22331/q-2022-10-13-840},
  title = {Quantifying fermionic interactions from the violation of {W}ick's theorem},
  author = {Pachos, Jiannis K. and Vlachou, Chrysoula},
  journal = {{Quantum}},
  issn = {2521-327X},
  publisher = {{Verein zur F{\"{o}}rderung des Open Access Publizierens in den Quantenwissenschaften}},
  volume = {6},
  pages = {840},
  month = oct,
  year = {2022}
}

@article{cqfh-mbwk,
  title = {Local classical correlations between physical electrons in Hubbard systems},
  author = {Bellomia, Gabriele and Amaricci, Adriano and Capone, Massimo},
  journal = {Phys. Rev. B},
  volume = {113},
  issue = {15},
  pages = {155158},
  numpages = {10},
  year = {2026},
  month = {Apr},
  publisher = {American Physical Society},
  doi = {10.1103/cqfh-mbwk},
  url = {https://link.aps.org/doi/10.1103/cqfh-mbwk}
}

@misc{zavatti2026quantummagicstronglycorrelated,
      title={Quantum magic of strongly correlated fermions $-$ the Hubbard dimer}, 
      author={Edoardo Zavatti and Gabriele Bellomia and Massimo Capone},
      year={2026},
      eprint={2605.18494},
      archivePrefix={arXiv},
      primaryClass={quant-ph},
      url={https://arxiv.org/abs/2605.18494}, 
}

@article{79vj-nx6r,
  title = {Growth and spreading of quantum resources under random circuit dynamics},
  author = {Aditya, Sreemayee and Turkeshi, Xhek and Sierant, Piotr},
  journal = {Phys. Rev. Res.},
  volume = {8},
  issue = {3},
  pages = {033062},
  numpages = {13},
  year = {2026},
  month = {Jul},
  publisher = {American Physical Society},
  doi = {10.1103/79vj-nx6r},
  url = {https://link.aps.org/doi/10.1103/79vj-nx6r}
}

@misc{matsuda2026quantumcomputationalresourcesconformal,
      title={Quantum Computational Resources and Conformal Field Theory: Unifying Spins, Bosons, and Fermions}, 
      author={Ryota Matsuda and Masahiro Hoshino and Yuto Ashida},
      year={2026},
      eprint={2607.05343},
      archivePrefix={arXiv},
      primaryClass={quant-ph},
      url={https://arxiv.org/abs/2607.05343}, 
}

@misc{ares2026asymmetrylowerboundfermionic,
      title={An asymmetry lower bound on fermionic non-Gaussianity}, 
      author={Filiberto Ares and Michele Mazzoni and Sara Murciano and Dávid Szász-Schagrin and Pasquale Calabrese and Lorenzo Piroli},
      year={2026},
      eprint={2603.16762},
      archivePrefix={arXiv},
      primaryClass={quant-ph},
      url={https://arxiv.org/abs/2603.16762}, 
}

@misc{aditya2025mpembaeffectsquantumcomplexity,
      title={Mpemba Effects in Quantum Complexity}, 
      author={Sreemayee Aditya and Alessandro Summer and Piotr Sierant and Xhek Turkeshi},
      year={2025},
      eprint={2509.22176},
      archivePrefix={arXiv},
      primaryClass={quant-ph},
      url={https://arxiv.org/abs/2509.22176}, 
}

@misc{ares2026nongaussianityrandomquantumstates,
      title={Non-Gaussianity of random quantum states}, 
      author={Filiberto Ares and Sara Murciano and Pasquale Calabrese},
      year={2026},
      eprint={2605.18986},
      archivePrefix={arXiv},
      primaryClass={cond-mat.stat-mech},
      url={https://arxiv.org/abs/2605.18986}, 
}

@misc{falcao2026fermionicmagicresourcesdisordered,
      title={Fermionic magic resources in disordered quantum spin chains}, 
      author={Pedro R. Nicácio Falcão and Jakub Zakrzewski and Piotr Sierant},
      year={2026},
      eprint={2602.00245},
      archivePrefix={arXiv},
      primaryClass={quant-ph},
      url={https://arxiv.org/abs/2602.00245}, 
}

@article{PhysRevResearch.6.023176,
  title = {Measurement-induced transitions beyond Gaussianity: A single particle description},
  author = {Lumia, Luca and Tirrito, Emanuele and Fazio, Rosario and Collura, Mario},
  journal = {Phys. Rev. Res.},
  volume = {6},
  issue = {2},
  pages = {023176},
  numpages = {10},
  year = {2024},
  month = {May},
  publisher = {American Physical Society},
  doi = {10.1103/PhysRevResearch.6.023176},
  url = {https://link.aps.org/doi/10.1103/PhysRevResearch.6.023176}
}

@misc{debertolis2025naturalsuperorbitalsrepresentationmanybody,
      title={Natural super-orbitals representation of many-body operators}, 
      author={Maxime Debertolis},
      year={2025},
      eprint={2507.10690},
      archivePrefix={arXiv},
      primaryClass={cond-mat.str-el},
      url={https://arxiv.org/abs/2507.10690}, 
}

@article{Bittel2025paclearningoffree,
  doi = {10.22331/q-2025-03-20-1665},
  url = {https://doi.org/10.22331/q-2025-03-20-1665},
  title = {{PAC}-learning of free-fermionic states is {NP}-hard},
  author = {Bittel, Lennart and Mele, Antonio A. and Eisert, Jens and Leone, Lorenzo},
  journal = {{Quantum}},
  issn = {2521-327X},
  publisher = {{Verein zur F{\"{o}}rderung des Open Access Publizierens in den Quantenwissenschaften}},
  volume = {9},
  pages = {1665},
  month = mar,
  year = {2025}
}

@article{PhysRevA.102.052604,
  title = {Computational power of matchgates with supplementary resources},
  author = {Hebenstreit, M. and Jozsa, R. and Kraus, B. and Strelchuk, S.},
  journal = {Phys. Rev. A},
  volume = {102},
  issue = {5},
  pages = {052604},
  numpages = {15},
  year = {2020},
  month = {Nov},
  publisher = {American Physical Society},
  doi = {10.1103/PhysRevA.102.052604},
  url = {https://link.aps.org/doi/10.1103/PhysRevA.102.052604}
}

@article{PhysRevLett.120.190501,
  title = {Fidelity Witnesses for Fermionic Quantum Simulations},
  author = {Gluza, M. and Kliesch, M. and Eisert, J. and Aolita, L.},
  journal = {Phys. Rev. Lett.},
  volume = {120},
  issue = {19},
  pages = {190501},
  numpages = {7},
  year = {2018},
  month = {May},
  publisher = {American Physical Society},
  doi = {10.1103/PhysRevLett.120.190501},
  url = {https://link.aps.org/doi/10.1103/PhysRevLett.120.190501}
}

@misc{santra2025quantumresourcesnonabelianlattice,
      title={Quantum Resources in Non-Abelian Lattice Gauge Theories: Nonstabilizerness, Multipartite Entanglement, and Fermionic Non-Gaussianity}, 
      author={Gopal Chandra Santra and Julius Mildenberger and Edoardo Ballini and Alberto Bottarelli and Matteo M. Wauters and Philipp Hauke},
      year={2025},
      eprint={2510.07385},
      archivePrefix={arXiv},
      primaryClass={quant-ph},
      url={https://arxiv.org/abs/2510.07385}, 
}

@misc{lyu2024fermionicgaussiantestingnongaussian,
      title={Fermionic Gaussian Testing and Non-Gaussian Measures via Convolution}, 
      author={Xingjian Lyu and Kaifeng Bu},
      year={2024},
      eprint={2409.08180},
      archivePrefix={arXiv},
      primaryClass={quant-ph},
      url={https://arxiv.org/abs/2409.08180}, 
}

@book{FetterWalecka03,
  title     = {Quantum Theory of Many-Particle Systems},
  author    = {Fetter, Alexander L. and Walecka, John Dirk},
  publisher = {Dover Publications},
  address   = {Mineola, NY},
  year      = {2003},
  note      = {Original edition: McGraw-Hill, 1971}
}

@book{Negele18quantum,
  title={Quantum many-particle systems},
  author={Negele, John W and Orland, Henri},
  year={2018},
  series={Frontiers in Physics},
  publisher={CRC Press}
}

@book{AltlandSimons10,
  title     = {Condensed Matter Field Theory},
  author    = {Altland, Alexander and Simons, Ben D.},
  edition   = {2},
  publisher = {Cambridge University Press},
  address   = {Cambridge},
  year      = {2010},
  doi       = {10.1017/CBO9780511789984}
}

@book{deGennes99,
  title     = {Superconductivity of Metals and Alloys},
  author    = {de Gennes, Pierre-Gilles},
  publisher = {Westview Press},
  address   = {Boulder, CO},
  year      = {1999},
  note      = {Original edition: W. A. Benjamin, 1966}
}

@article{Schrieffer18,
  title = {Theory of Superconductivity},
  author = {Bardeen, J. and Cooper, L. N. and Schrieffer, J. R.},
  journal = {Phys. Rev.},
  volume = {108},
  issue = {5},
  pages = {1175--1204},
  numpages = {0},
  year = {1957},
  month = {Dec},
  publisher = {American Physical Society},
  doi = {10.1103/PhysRev.108.1175},
  url = {https://link.aps.org/doi/10.1103/PhysRev.108.1175}
}

@article{ChitambarGour2019,
  author  = {Chitambar, Eric and Gour, Gilad},
  title   = {Quantum resource theories},
  journal = {Rev. Mod. Phys.},
  volume  = {91},
  pages   = {025001},
  year    = {2019},
  doi     = {10.1103/RevModPhys.91.025001},
}

@article{Nielsen1999,
  author  = {Nielsen, Michael A.},
  title   = {Conditions for a class of entanglement transformations},
  journal = {Phys. Rev. Lett.},
  volume  = {83},
  pages   = {436},
  year    = {1999},
  doi     = {10.1103/PhysRevLett.83.436},
}

@article{NielsenVidal2001,
  author  = {Nielsen, Michael A. and Vidal, Guifr{\'e}},
  title   = {Majorization and the interconversion of bipartite states},
  journal = {Quantum Inf. Comput.},
  volume  = {1},
  pages   = {76},
  year    = {2001},
}

@article{JonathanPlenio1999a,
  author  = {Jonathan, Daniel and Plenio, Martin B.},
  title   = {Minimal conditions for local pure-state entanglement manipulation},
  journal = {Phys. Rev. Lett.},
  volume  = {83},
  pages   = {1455},
  year    = {1999},
  doi     = {10.1103/PhysRevLett.83.1455},
}

@article{Vidal1999,
  author  = {Vidal, Guifr{\'e}},
  title   = {Entanglement of pure states for a single copy},
  journal = {Phys. Rev. Lett.},
  volume  = {83},
  pages   = {1046},
  year    = {1999},
  doi     = {10.1103/PhysRevLett.83.1046},
}

@article{PopescuRohrlich1997,
  author  = {Popescu, Sandu and Rohrlich, Daniel},
  title   = {Thermodynamics and the measure of entanglement},
  journal = {Phys. Rev. A},
  volume  = {56},
  pages   = {R3319},
  year    = {1997},
  doi     = {10.1103/PhysRevA.56.R3319},
}

@article{DonaldHorodeckiRudolph2002,
  author  = {Donald, Matthew J. and Horodecki, Micha{\l} and Rudolph, Oliver},
  title   = {The uniqueness theorem for entanglement measures},
  journal = {J. Math. Phys.},
  volume  = {43},
  pages   = {4252},
  year    = {2002},
  doi     = {10.1063/1.1495917},
}

@article{JonathanPlenio1999b,
  author  = {Jonathan, Daniel and Plenio, Martin B.},
  title   = {Entanglement-assisted local manipulation of pure quantum states},
  journal = {Phys. Rev. Lett.},
  volume  = {83},
  pages   = {3566},
  year    = {1999},
  doi     = {10.1103/PhysRevLett.83.3566},
}

@article{Turgut2007,
  author  = {Turgut, S.},
  title   = {Catalytic transformations for bipartite pure states},
  journal = {J. Phys. A: Math. Theor.},
  volume  = {40},
  pages   = {12185},
  year    = {2007},
  doi     = {10.1088/1751-8113/40/40/012},
}

@article{Klimesh2007,
  author  = {Klimesh, M.},
  title   = {Inequalities that collectively completely characterize the catalytic majorization relation},
  journal = {arXiv:0709.3680},
  year    = {2007},
}

@article{BrodGalvao2012,
  author  = {Brod, Daniel J. and Galv\~{a}o, Ernesto F.},
  title   = {Geometries for universal quantum computation with matchgates},
  journal = {Phys. Rev. A},
  volume  = {86},
  pages   = {052307},
  year    = {2012},
  doi     = {10.1103/PhysRevA.86.052307},
}

@article{BennettEtAl1996,
  author  = {Bennett, Charles H. and Bernstein, Herbert J. and Popescu, Sandu and Schumacher, Benjamin},
  title   = {Concentrating partial entanglement by local operations},
  journal = {Phys. Rev. A},
  volume  = {53},
  pages   = {2046},
  year    = {1996},
  doi     = {10.1103/PhysRevA.53.2046},
}

@article{VidalCirac2001,
  author  = {Vidal, Guifr{\'e} and Cirac, J. Ignacio},
  title   = {Irreversibility in asymptotic manipulations of entanglement},
  journal = {Phys. Rev. Lett.},
  volume  = {86},
  pages   = {5803},
  year    = {2001},
  doi     = {10.1103/PhysRevLett.86.5803},
}

@article{bera2015mbl,
  title = {Many-Body Localization Characterized from a One-Particle Perspective},
  author = {Bera, Soumya and Schomerus, Henning and Heidrich-Meisner, Fabian and Bardarson, Jens H.},
  journal = {Physical Review Letters},
  volume = {115},
  issue = {4},
  pages = {046603},
  numpages = {5},
  year = {2015},
  month = jul,
  publisher = {American Physical Society},
  doi = {10.1103/PhysRevLett.115.046603},
  url = {https://link.aps.org/doi/10.1103/PhysRevLett.115.046603}
}

@article{AmicoFazioOsterlohVedral2008,
  author  = {Amico, Luigi and Fazio, Rosario and Osterloh, Andreas and Vedral, Vlatko},
  title   = {Entanglement in many-body systems},
  journal = {Rev. Mod. Phys.},
  volume  = {80},
  pages   = {517},
  year    = {2008},
  doi     = {10.1103/RevModPhys.80.517},
}

@article{HorodeckiEtAl2009,
  author  = {Horodecki, Ryszard and Horodecki, Pawe{\l} and Horodecki, Micha{\l} and Horodecki, Karol},
  title   = {Quantum entanglement},
  journal = {Rev. Mod. Phys.},
  volume  = {81},
  pages   = {865},
  year    = {2009},
  doi     = {10.1103/RevModPhys.81.865},
}

@misc{sierant2026theorymatchgatecommutant,
      title={Theory of the Matchgate Commutant}, 
      author={Piotr Sierant and Xhek Turkeshi and Poetri Sonya Tarabunga},
      year={2026},
      eprint={2603.12392},
      archivePrefix={arXiv},
      url={https://arxiv.org/abs/2603.12392}, 
}

@misc{braccia2026commutant,
      title={The commutant of fermionic Gaussian unitaries}, 
      author={Paolo Braccia and N. L. Diaz and Martin Larocca and M. Cerezo and Diego García-Martín},
      year={2026},
      eprint={2603.19210},
      archivePrefix={arXiv},
      primaryClass={quant-ph},
      url={https://arxiv.org/abs/2603.19210}, 
}

@article{Swietek26one,
  title={One-Body Purity, Non-Gaussianity, and Entanglement in Interacting Integrable Models},
  author={{\'S}wi{\k{e}}tek, Rafa{\l} and Kliczkowski, Maksymilian and Vidmar, Lev and Rigol, Marcos},
  journal={arXiv preprint arXiv:2607.01326},
  year={2026}
}

@misc{lastres2026geometry,
      title={Geometry of Free Fermion Commutants}, 
      author={Marco Lastres and Sanjay Moudgalya},
      year={2026},
      eprint={2604.05031},
      archivePrefix={arXiv},
      primaryClass={quant-ph},
      url={https://arxiv.org/abs/2604.05031}, 
}

@article{TirritoTurkeshiSierant2024,
  title = {Anticoncentration and Nonstabilizerness Spreading under Ergodic Quantum Dynamics},
  author = {Tirrito, Emanuele and Turkeshi, Xhek and Sierant, Piotr},
  journal = {Phys. Rev. Lett.},
  volume = {135},
  issue = {22},
  pages = {220401},
  numpages = {9},
  year = {2025},
  month = {Nov},
  publisher = {American Physical Society},
  doi = {10.1103/1jzy-sk9r},
  url = {https://link.aps.org/doi/10.1103/1jzy-sk9r}
}

@misc{Greenberger07,
      title={Going Beyond Bell's Theorem}, 
      author={Daniel M. Greenberger and Michael A. Horne and Anton Zeilinger},
      year={2007},
      eprint={0712.0921},
      archivePrefix={arXiv},
      primaryClass={quant-ph},
      url={https://arxiv.org/abs/0712.0921}, 
}

@article{StreltsovAdessoPlenio2017,
  author  = {Streltsov, Alexander and Adesso, Gerardo and Plenio, Martin B.},
  title   = {Colloquium: Quantum coherence as a resource},
  journal = {Rev. Mod. Phys.},
  volume  = {89},
  pages   = {041003},
  year    = {2017},
  doi     = {10.1103/RevModPhys.89.041003},
}

@article{DuBaiGuo2015,
  author  = {Du, Shuanping and Bai, Zhaofang and Guo, Yu},
  title   = {Conditions for coherence transformations under incoherent operations},
  journal = {Phys. Rev. A},
  volume  = {91},
  pages   = {052120},
  year    = {2015},
  doi     = {10.1103/PhysRevA.91.052120},
  note    = {Erratum: Phys. Rev. A \textbf{95}, 029901 (2017)},
}

@article{HorodeckiOppenheim2013,
  author  = {Horodecki, Micha{\l} and Oppenheim, Jonathan},
  title   = {Fundamental limitations for quantum and nanoscale thermodynamics},
  journal = {Nat. Commun.},
  volume  = {4},
  pages   = {2059},
  year    = {2013},
  doi     = {10.1038/ncomms3059},
}

@article{BrandaoEtAl2015,
  author  = {Brand{\~a}o, Fernando and Horodecki, Micha{\l} and Ng, Nelly and
             Oppenheim, Jonathan and Wehner, Stephanie},
  title   = {The second laws of quantum thermodynamics},
  journal = {Proc. Natl. Acad. Sci. U.S.A.},
  volume  = {112},
  pages   = {3275},
  year    = {2015},
  doi     = {10.1073/pnas.1411728112},
}

@article{GourEtAl2015,
  author  = {Gour, Gilad and M{\"u}ller, Markus P. and Narasimhachar, Varun and
             Spekkens, Robert W. and Yunger Halpern, Nicole},
  title   = {The resource theory of informational nonequilibrium in thermodynamics},
  journal = {Phys. Rep.},
  volume  = {583},
  pages   = {1},
  year    = {2015},
  doi     = {10.1016/j.physrep.2015.04.003},
  eprint  = {1309.6586},
}

@article{Valiant2002,
  author  = {Valiant, Leslie G.},
  title   = {Quantum circuits that can be simulated classically in polynomial time},
  journal = {SIAM J. Comput.},
  volume  = {31},
  pages   = {1229},
  year    = {2002},
  doi     = {10.1137/S0097539700377025},
}

@article{TerhalDiVincenzo2002,
  author  = {Terhal, Barbara M. and DiVincenzo, David P.},
  title   = {Classical simulation of noninteracting-fermion quantum circuits},
  journal = {Phys. Rev. A},
  volume  = {65},
  pages   = {032325},
  year    = {2002},
  doi     = {10.1103/PhysRevA.65.032325},
}

@article{Knill2001,
  author = {Knill, Emanuel},
  title  = {Fermionic linear optics and matchgates},
  year   = {2001},
  eprint = {quant-ph/0108033},
  journal = {arXiv:quant-ph/0108033},
}

@article{JozsaMiyake2008,
  author  = {Jozsa, Richard and Miyake, Akimasa},
  title   = {Matchgates and classical simulation of quantum circuits},
  journal = {Proc. R. Soc. A},
  volume  = {464},
  pages   = {3089},
  year    = {2008},
  doi     = {10.1098/rspa.2008.0189},
}

@article{Bravyi2005,
  author  = {Bravyi, Sergey},
  title   = {Lagrangian representation for fermionic linear optics},
  journal = {Quantum Inf. Comput.},
  volume  = {5},
  pages   = {216},
  year    = {2005},
  eprint  = {quant-ph/0404180},
}

@article{bravyi2002fermionic,
  title = {Fermionic Quantum Computation},
  volume = {298},
  ISSN = {0003-4916},
  url = {http://dx.doi.org/10.1006/aphy.2002.6254},
  DOI = {10.1006/aphy.2002.6254},
  number = {1},
  journal = {Annals of Physics},
  publisher = {Elsevier BV},
  author = {Bravyi,  Sergey B. and Kitaev,  Alexei Yu.},
  year = {2002},
  month = May,
  pages = {210–226}
}

@article{DiVincenzoTerhal2005,
  author  = {DiVincenzo, David P. and Terhal, Barbara M.},
  title   = {Fermionic linear optics revisited},
  journal = {Found. Phys.},
  volume  = {35},
  pages   = {1967},
  year    = {2005},
  doi     = {10.1007/s10701-005-8657-0},
}

@misc{trigueros2026unitarydesigns,
      title={Unitary Designs from Doped Matchgate Circuits}, 
      author={Fabian Ballar Trigueros and Zheng-Hang Sun and Xhek Turkeshi and Piotr Sierant and Poetri Sonya Tarabunga},
      year={2026},
      eprint={2606.23800},
      archivePrefix={arXiv},
      primaryClass={quant-ph},
      url={https://arxiv.org/abs/2606.23800}, 
}

@article{PeschelEisler2009,
  author  = {Peschel, Ingo and Eisler, Viktor},
  title   = {Reduced density matrices and entanglement entropy in free lattice models},
  journal = {J. Phys. A},
  volume  = {42},
  pages   = {504003},
  year    = {2009},
  doi     = {10.1088/1751-8113/42/50/504003},
}

@article{SuraceTagliacozzo2022,
  author  = {Surace, Jacopo and Tagliacozzo, Luca},
  title   = {Fermionic {G}aussian states: an introduction to numerical approaches},
  journal = {SciPost Phys. Lect. Notes},
  volume  = {54},
  pages   = {1},
  year    = {2022},
  doi     = {10.21468/SciPostPhysLectNotes.54},
  eprint  = {2111.08343},
}

@article{Oszmaniec22,
  title = {Fermion Sampling: A Robust Quantum Computational Advantage Scheme Using Fermionic Linear Optics and Magic Input States},
  author = {Oszmaniec, Micha\l{} and Dangniam, Ninnat and Morales, Mauro E.S. and Zimbor\'as, Zolt\'an},
  journal = {PRX Quantum},
  volume = {3},
  issue = {2},
  pages = {020328},
  numpages = {54},
  year = {2022},
  month = {May},
  publisher = {American Physical Society},
  doi = {10.1103/PRXQuantum.3.020328},
  url = {https://link.aps.org/doi/10.1103/PRXQuantum.3.020328}
}

@article{HebenstreitJozsaKrausStrelchuk2019,
  author  = {Hebenstreit, Martin and Jozsa, Richard and Kraus, Barbara and Strelchuk, Sergii and Yoganathan, Mithuna},
  title   = {All pure fermionic non-{G}aussian states are magic states for matchgate computations},
  journal = {Phys. Rev. Lett.},
  volume  = {123},
  pages   = {080503},
  year    = {2019},
  doi     = {10.1103/PhysRevLett.123.080503},
}

@article{DiasKoenig2024,
  author  = {Dias, Beatriz and K{\"o}nig, Robert},
  title   = {Classical simulation of non-{G}aussian fermionic circuits},
  journal = {Quantum},
  volume  = {8},
  pages   = {1350},
  year    = {2024},
  doi     = {10.22331/q-2024-05-21-1350},
}

@article{ReardonSmithOszmaniecKorzekwa2024,
  author  = {Reardon-Smith, Oliver and Oszmaniec, Micha{\l} and Korzekwa, Kamil},
  title   = {Improved simulation of quantum circuits dominated by free fermionic operations},
  journal = {Quantum},
  volume  = {8},
  pages   = {1549},
  year    = {2024},
  doi     = {10.22331/q-2024-12-04-1549},
}

@article{GigenaDiTullioRossignoli2020,
  author  = {Gigena, Nicol{\'a}s and Di Tullio, Marco and Rossignoli, Ra{\'u}l},
  title   = {One-body entanglement as a quantum resource in fermionic systems},
  journal = {Phys. Rev. A},
  volume  = {102},
  pages   = {042410},
  year    = {2020},
  doi     = {10.1103/PhysRevA.102.042410},
}

@article{SierantStornatiTurkeshi2026,
  author  = {Sierant, Piotr and Stornati, Paolo and Turkeshi, Xhek},
  title   = {Fermionic magic resources of quantum many-body systems},
  journal = {PRX Quantum},
  volume  = {7},
  pages   = {010302},
  year    = {2026},
  doi     = {10.1103/3yx4-1j27},
}

@article{TarabungaEtAl2026,
  author = {Tarabunga, Poetri S. and Jobst, Bernhard and Morral-Yepes, Rafael and Langer, M. and Kraus, Barbara and Pollmann, Frank and Lin, Sheng-Hsuan},
  title  = {Computable fermionic non-{G}aussianity from the covariance matrix},
  year   = {2026},
  eprint = {2607.02242},
  journal = {arXiv:2607.02242},
}

@article{HaugTurkeshiSierant2026,
  author = {Haug, Tobias and Turkeshi, Xhek and Sierant, Piotr},
  title  = {Practical tests and witnesses of fermionic non-{G}aussianity},
  year   = {2026},
  eprint = {2605.26218},
  journal = {arXiv:2605.26218},
}

@article{LeoneBittel2026,
  author = {Leone, Lorenzo and Bittel, Lennart},
  title  = {Fermionic entropy: an efficiently measurable strong monotone for non-{G}aussianity},
  year   = {2026},
  eprint = {2607.29670},
  journal = {arXiv:2607.29670},
}

@book{Bhatia1997,
  author    = {Bhatia, Rajendra},
  title     = {Matrix Analysis},
  publisher = {Springer},
  address   = {New York},
  year      = {1997},
  doi       = {10.1007/978-1-4612-0653-8},
  isbn      = {978-0-387-94846-1},
}

@book{MarshallOlkinArnold2011,
  author    = {Marshall, Albert W. and Olkin, Ingram and Arnold, Barry C.},
  title     = {Inequalities: Theory of Majorization and Its Applications},
  edition   = {2nd},
  publisher = {Springer},
  address   = {New York},
  year      = {2011},
  doi       = {10.1007/978-0-387-68276-1},
  isbn      = {978-0-387-40087-7},
}

@article{Williamson1936,
  author  = {Williamson, John},
  title   = {On the algebraic problem concerning the normal forms of linear dynamical systems},
  journal = {Am. J. Math.},
  volume  = {58},
  pages   = {141},
  year    = {1936},
  doi     = {10.2307/2371062},
}

@article{MeleHerasymenko2025,
  author  = {Mele, Antonio A. and Herasymenko, Yaroslav},
  title   = {Efficient learning of quantum states prepared with few fermionic non-{G}aussian gates},
  journal = {PRX Quantum},
  volume  = {6},
  pages   = {010319},
  year    = {2025},
  doi     = {10.1103/PRXQuantum.6.010319},
  eprint  = {2402.18665},
}

@article{TakagiZhuang2018,
  author  = {Takagi, Ryuji and Zhuang, Quntao},
  title   = {Convex resource theory of non-{G}aussianity},
  journal = {Phys. Rev. A},
  volume  = {97},
  pages   = {062337},
  year    = {2018},
  doi     = {10.1103/PhysRevA.97.062337},
}

@article{AlbarelliGenoniSerafiniFerraro2018,
  title = {Resource theory of quantum non-Gaussianity and Wigner negativity},
  author = {Albarelli, Francesco and Genoni, Marco G. and Paris, Matteo G. A. and Ferraro, Alessandro},
  journal = {Phys. Rev. A},
  volume = {98},
  issue = {5},
  pages = {052350},
  numpages = {17},
  year = {2018},
  month = {Nov},
  publisher = {American Physical Society},
  doi = {10.1103/PhysRevA.98.052350},
  url = {https://link.aps.org/doi/10.1103/PhysRevA.98.052350}
}

@book{RockafellarConvex,
  author    = {Rockafellar, R. Tyrrell},
  title     = {Convex Analysis},
  series    = {Princeton Mathematical Series},
  volume    = {28},
  publisher = {Princeton University Press},
  address   = {Princeton, NJ},
  year      = {1970}
}

@article{LamiEtAl2018,
  author  = {Lami, Ludovico and Regula, Bartosz and Wang, Xin and Nichols, Rosanna and Winter, Andreas and Adesso, Gerardo},
  title   = {Gaussian quantum resource theories},
  journal = {Phys. Rev. A},
  volume  = {98},
  pages   = {022335},
  year    = {2018},
  doi     = {10.1103/PhysRevA.98.022335},
}

@article{PaviglianitiEtAl2026,
  author  = {Paviglianiti, Alessio and Lumia, Luca and Tirrito, Emanuele and Silva, Alessandro and Collura, Mario and Turkeshi, Xhek and Lami, Guglielmo},
  title   = {Emergence of generic entanglement structure in doped matchgate circuits},
  journal = {Phys. Rev. Lett.},
  volume  = {136},
  pages   = {020403},
  year    = {2026},
  doi     = {10.1103/w97w-7zny},
}

@article{tarabunga2026fermionic,
      title={Fermionic non-Gaussianity via Bell sampling: monotones and efficient quantum algorithms}, 
      author={Poetri Sonya Tarabunga},
      year={2026},
      journal={arXiv:2606.05066},
      url={https://arxiv.org/abs/2606.05066}, 
}

@article{bittel2025optimal,
      title = {Optimal Trace-Distance Bounds for Free-Fermionic States: Testing and Improved Tomography},
  volume = {6},
  pages = {030341},
  ISSN = {2691-3399},
  url = {http://dx.doi.org/10.1103/pzx6-nkfb},
  DOI = {10.1103/pzx6-nkfb},
  number = {3},
  journal = {PRX Quantum},
  publisher = {American Physical Society (APS)},
  author = {Bittel,  Lennart and Mele,  Antonio Anna and Eisert,  Jens and Leone,  Lorenzo},
  year = {2025},
  month = sep 
}

@misc{SuppMat,
  note = {See Supplemental Material appended to this manuscript for the complete proofs of all statements},
}

@article{mele2026symplectic,
  title = {Symplectic Rank of Non-Gaussian Quantum States},
  author = {Mele, Francesco A. and Oliviero, Salvatore F.E. and Upreti, Varun and Chabaud, Ulysse},
  journal = {PRX Quantum},
  volume = {7},
  issue = {2},
  pages = {020366},
  numpages = {58},
  year = {2026},
  month = {Jun},
  publisher = {American Physical Society},
  doi = {10.1103/1rtk-1jsn},
  url = {https://link.aps.org/doi/10.1103/1rtk-1jsn}
}

@article{Beverland2020,
  author  = {Beverland, Michael and Campbell, Earl and Howard, Mark and Kliuchnikov, Vadym},
  title   = {Lower bounds on the non-Clifford resources for quantum computations},
  journal = {Quantum Sci. Technol.},
  volume  = {5},
  pages   = {035009},
  year    = {2020},
  doi     = {10.1088/2058-9565/ab8963},
}

@article{KondraDattaStreltsov2021,
  title = {Catalytic Transformations of Pure Entangled States},
  author = {Kondra, Tulja Varun and Datta, Chandan and Streltsov, Alexander},
  journal = {Phys. Rev. Lett.},
  volume = {127},
  issue = {15},
  pages = {150503},
  numpages = {6},
  year = {2021},
  month = {Oct},
  publisher = {American Physical Society},
  doi = {10.1103/PhysRevLett.127.150503},
  url = {https://link.aps.org/doi/10.1103/PhysRevLett.127.150503}
}

\clearpage
\onecolumngrid
\setcounter{secnumdepth}{4}
\begin{center}
{\large\bfseries Supplemental Material for\\[4pt] \titleinfo}\\[10pt]
Xhek Turkeshi, Piotr Sierant, and Poetri Sonya Tarabunga
\end{center}
\vspace{6pt}

\setcounter{equation}{0}
\setcounter{section}{0}
\setcounter{lemma}{0}
\renewcommand{\theequation}{S\arabic{equation}}
\renewcommand{\thesection}{S\arabic{section}}
\renewcommand{\thelemma}{S\arabic{lemma}}

\noindent
This Supplemental Material contains the complete proofs of all statements
of the main text. Section~\ref{sm:conventions} collects the covariance
conventions and elementary identities; Sec.~\ref{sm:witness} develops the
variational theory of Williamson partial sums; Sec.~\ref{sm:localization}
proves the localization of an optimal witness around the measured plane;
Sec.~\ref{sm:code} establishes the Gaussian-code inequality;
Sec.~\ref{sm:proof-thm1} assembles the proof of Theorem~\ref{thm:main},
including the lift from a single occupation measurement to arbitrary
Gaussian protocols; Sec.~\ref{sm:mixed} presents the mixed-state generalization of Theorem~\ref{thm:main} and the role of recorded versus unrecorded mode discards;
Sec.~\ref{sm:thm2} proves Theorem~\ref{thm:admissible}, including the
converse construction showing that convexity of the generator is
necessary; Sec.~\ref{sm:nullity} gives a direct, self-contained proof that
the Gaussian nullity is a strong monotone; and Sec.~\ref{sm:asymptotic}
derives the asymptotic conversion bounds, Corollary~\ref{cor:rate} and
Eq.~\eqref{eq:two-directions}, and completes the proof of
Theorem~\ref{thm:irrev}.

\section{Conventions and elementary covariance identities}
\label{sm:conventions}

\subsection{Physicality, canonical form, and faithfulness}

Throughout, $\Tr$ denotes the Hilbert-space trace and $\tr$ the trace of
matrices on the real Majorana space $V\simeq\mathbb R^{2N}$. For
$z\in\mathbb C^{2N}$, define $A_z=\sum_az_a\gamma_a$. Positivity of
$\Tr(\rho A_z^\dagger A_z)$ gives
\begin{equation}
  z^\dagger(\id+i\Gam_\rho)z\ge0,
  \qquad\text{hence}\qquad
  \id+i\Gam_\rho\ge0 .
  \label{eq:sm-physicality}
\end{equation}
Since $i\Gam_\rho$ has eigenvalues $\pm r_j$, this implies $0\le r_j\le1$
for every $j$. Every real antisymmetric matrix admits the orthogonal
canonical form
\begin{equation}
  O\Gam O^T=\bigoplus_{j=1}^{N}
  \begin{pmatrix}0&r_j\\-r_j&0\end{pmatrix},
  \qquad O\in O(2N),
  \label{eq:sm-canonical}
\end{equation}
with the singular values counted once per two-plane~\cite{Williamson1936,Bhatia1997}.
For a pure state, $\Gam^T\Gam\le\id$ with equality if
and only if $\Gam^2=-\id$; in that case 
the state is Gaussian~\cite{Bravyi2005,SuraceTagliacozzo2022}. Hence, for
pure $\psi$,
\begin{equation}
  r_1=\cdots=r_N=1
  \quad\Longleftrightarrow\quad
  \psi\ \text{is pure Gaussian},
  \label{eq:sm-faithfulness}
\end{equation}
the faithfulness fact used in Theorem~\ref{thm:admissible}(iii). On the other hand, no such
statement holds on mixed states: the maximally mixed state is Gaussian
with $\Gam=0$, so raw spectral deficits do not vanish on mixed Gaussian
states.
A generalization of Theorem~\ref{thm:main} for mixed state is then possible only with some restrictions, cf.~ Sec.~\ref{sm:mixed}.

\subsection{The affine identity for the unmeasured block}
\label{sm:affine}

Let us order the measured mode first and split the Majorana space as
$V=A\oplus B$ with $\dim A=2$. Let $\Pi_s=\ketbra{s}{s}_A$, $s\in\{0,1\}$,
be the occupation projectors of the measured mode, and
\begin{equation}
  p_s=\Tr(\Pi_s\rho),
  \qquad
  \rho_s=\frac{\langle s|\rho|s\rangle}{p_s}
  \quad (p_s>0),
  \label{eq:sm-conditional}
\end{equation}
with $p_s=0$ terms omitted throughout. In the Jordan--Wigner form
$\gamma_1=X\otimes\id$, $\gamma_2=Y\otimes\id$,
$\gamma_{a+2}=Z\otimes\mu_a$, with $\mu_a$ the intrinsic Majoranas of $B$,
every unmeasured bilinear has its two parity strings cancel,
\begin{equation}
  X_{ab}:=-\frac{i}{2}[\gamma_{a+2},\gamma_{b+2}]
  =\id_A\otimes\Bigl(-\frac{i}{2}[\mu_a,\mu_b]\Bigr)
  =:\id_A\otimes X^{(B)}_{ab},
  \label{eq:sm-strings}
\end{equation}
so $\Pi_0X_{ab}\Pi_1=\Pi_1X_{ab}\Pi_0=0$. Using
$\Pi_s\rho\Pi_s=\ketbra{s}{s}\otimes p_s\rho_s$,
\begin{equation}
  (\Gam_\rho|_B)_{ab}
  =\Tr(\rho X_{ab})
  =\sum_{s}\Tr(\Pi_s\rho\Pi_sX_{ab})
  =\sum_sp_s(\Gam_{\rho_s}')_{ab},
  \label{eq:sm-affine}
\end{equation}
where $\Gam_{\rho_s}'$ denotes the branch covariance on $B$. Two remarks
are in order. First, Eq.~\eqref{eq:sm-affine} is an identity for matrix
entries, not for spectra: in general
$r_j(\sum_sp_s\Gam'_{\rho_s})\neq\sum_sp_sr_j(\Gam'_{\rho_s})$, and
controlling this nonlinearity is the purpose of the majorization argument.
Second, no condition is placed on the off-diagonal blocks
$\Pi_0\rho\Pi_1$; they generate the parent $A$--$B$ correlations and are
the reason the affine identity alone proves nothing.

For a product $\rho\otimes\sigma$, cross-covariance entries factor into
odd Majorana expectations of the factors. If either factor has definite
fermionic parity, its odd expectations vanish and
\begin{equation}
  \Gam_{\rho\otimes\sigma}=\Gam_\rho\oplus\Gam_\sigma,
  \label{eq:sm-concatenation}
\end{equation}
so the Williamson spectra concatenate,
$\rvec(\Gam_{\rho\otimes\sigma})=\rvec(\Gam_\rho)\sqcup\rvec(\Gam_\sigma)$.
Pure Gaussian ancillas and recorded occupation states have definite
parity. Spectra of different
lengths are compared after padding by unit entries to a common length; this corresponds to appending with a pure Gaussian ancilla on the deficit modes.

\section{Williamson partial sums as quadratic witnesses}
\label{sm:witness}

An $\ell$-plane \emph{partial complex structure} on $V$ is a real matrix
$J$ with
\begin{equation}
  J^T=-J,\qquad -J^2=P_W,\qquad \dim W=2\ell,
  \label{eq:sm-pcs}
\end{equation}
where $P_W$ projects onto $W=\Ran J$; we denote the set of such matrices
by $\cJ_\ell(V)$. To each $J$ we associate the quadratic observable
\begin{equation}
  Q_J:=-\frac{i}{2}\sum_{a,b=1}^{2N}J_{ab}\gamma_a\gamma_b,
  \qquad
  \Tr(\rho Q_J)=\frac12\tr(J^T\Gam_\rho),
  \label{eq:sm-QJ}
\end{equation}
where the expectation identity follows by inserting the definition of
$\Gam_\rho$ and antisymmetry of $J$.

\begin{lemma}[Variational formula]
\label{lem:sm-kyfan}
For every real antisymmetric $\Gam$,
\begin{equation}
  S_\ell(\Gam)=\max_{J\in\cJ_\ell(V)}\frac12\tr(J^T\Gam).
  \label{eq:sm-kyfan}
\end{equation}
\end{lemma}

\begin{proof}
Von Neumann's trace inequality bounds $|\tr(A^TB)|$ by the sum of products
of ordered singular values~\cite{Bhatia1997}. The singular values of $J$
are $1$, repeated $2\ell$ times, and $0$ otherwise; those of $\Gam$ are
the doubled Williamson values $r_1,r_1,\ldots,r_N,r_N$. Therefore
$\frac12\tr(J^T\Gam)\le\frac12\sum_{a\le2\ell}s_a(\Gam)=\sum_{j\le\ell}r_j$.
Equality is attained by taking $J$ to coincide, in the canonical basis of
Eq.~\eqref{eq:sm-canonical}, with the canonical complex structure on the
first $\ell$ planes, with matching orientations, and to vanish elsewhere.
\end{proof}

In an orthonormal basis adapted to $J$, the witness decomposes as
\begin{equation}
  Q_J=g_1+\cdots+g_\ell,
  \qquad
  g_j=-i\widetilde\gamma_{2j-1}\widetilde\gamma_{2j},
  \label{eq:sm-reflections}
\end{equation}
where the $g_j$ are commuting Hermitian reflections, $g_j^2=\id$, since
their Majorana supports are disjoint. This is the conceptual pivot of the
proof: once an optimizer is fixed, the nonlinear spectral quantity
$S_\ell$ has become the expectation of a \emph{linear} observable.

We stress what the variational formula does \emph{not} do: the witness may
be expressed in any convenient Majorana basis, but the physical
measurement is not rotated along with it---a general Bogoliubov rotation
would turn an occupation measurement into a pairing measurement. The
localization below keeps the physical measured plane $A$ fixed and
decomposes the witness relative to the split $A\oplus B$.

\section{Localization of an optimal witness}
\label{sm:localization}

\begin{lemma}[Localization]
\label{lem:sm-localization}
Let $V=A\oplus B$ with $\dim A=2$, and let $J\in\cJ_\ell(V)$. There exist
$r\in\{0,1,2\}$ and an orthogonal decomposition
$J=J_{\mathrm{loc}}\oplus J_c$ such that $J_c$ contains $\ell-r$ complex
planes supported entirely in $B$, while $J_{\mathrm{loc}}$ contains the
remaining $r$ planes and is supported on $A\oplus B_{\mathrm{loc}}$ with
$B_{\mathrm{loc}}\subseteq B$, $\dim B_{\mathrm{loc}}\le4$.
\end{lemma}

\begin{proof}
Let $W=\Ran J$ and $P_W=-J^2$, and set
\begin{equation}
  L=\operatorname{span}\{P_Wa,\ JP_Wa:\ a\in A\}.
  \label{eq:sm-L}
\end{equation}
The space $L$ is $J$-invariant, because on $W$ one has $J(JP_Wa)=-P_Wa$,
and a $J$-invariant real subspace has even dimension; since
$\dim P_WA\le2$, we get $\dim L=2r$ with $r\in\{0,1,2\}$. Define
$W_c=W\cap L^\perp$. It is $J$-invariant---for $w\in W_c$ and
$\ell_0\in L$, $\langle Jw,\ell_0\rangle=-\langle w,J\ell_0\rangle=0$---and
it lies in $B$: for $a\in A$,
$\langle w,a\rangle=\langle w,P_Wa\rangle=0$ because $w\in W$ and
$P_Wa\in L$. Thus $J_c:=J|_{W_c}$ never touches the measured plane. Let
$D=P_BL\subset B$; then $D\perp W_c$ and $\dim D\le2r\le4$. If $D$ is odd
dimensional, then $\dim D<2r$ and one direction from
$(B\cap W_c^\perp)\cap D^\perp$ may be added; the resulting
even-dimensional $B_{\mathrm{loc}}$ contains $D$, is orthogonal to $W_c$,
and satisfies $\dim B_{\mathrm{loc}}\le4$. Since $L\subset A\oplus D$, the
restriction $J_{\mathrm{loc}}:=J|_L$ is supported on
$A\oplus B_{\mathrm{loc}}$.
\end{proof}

The subspace $B_{\mathrm{loc}}$ is compatible with a physical mode
structure: choose $O_B\in SO(B)$ mapping $B_{\mathrm{loc}}$ to a canonical
pair of fermionic modes and let $U_B$ be the corresponding Gaussian
unitary on the unmeasured modes. Since $U_B$ commutes with the measurement
of $A$, the reduced state
\begin{equation}
  \sigma=\Tr_{B\ominus B_{\mathrm{loc}}}
  \bigl[(\id_A\otimes U_B)\rho(\id_A\otimes U_B^\dagger)\bigr]
  \label{eq:sm-reduced}
\end{equation}
has branch states
$\sigma_s=\Tr_{B\ominus B_{\mathrm{loc}}}[U_B\rho_sU_B^\dagger]$. We emphasize that even for a pure
global input, $\sigma$ is generally \emph{mixed}, because spectator modes
are traced out; every local inequality below must therefore hold---and
does hold---for arbitrary density matrices.

\section{The Gaussian-code inequality}
\label{sm:code}

Only the case $r=2$ of the localization is nontrivial, and it is here that
the fermionic structure enters. Let the local system carry $m$ fermionic
modes and let $J\in\cJ_2$, so that $Q_J=g_1+g_2$ with commuting quadratic
reflections. Define the joint $+1$ code projector
\begin{equation}
  P_{++}:=\frac14(\id+g_1)(\id+g_2),
  \qquad
  \rank P_{++}=2^{m-2}.
  \label{eq:sm-Ppp}
\end{equation}
Comparing the joint eigenvalues in the four sectors $++,+-,-+,--$, where
$Q_J$ takes the values $\{2,0,0,-2\}$ and $2P_{++}$ the values $\{2,0,0,0\}$,
gives the operator inequality
\begin{equation}
  Q_J=g_1+g_2\le2P_{++}.
  \label{eq:sm-QleP}
\end{equation}
The normalized projector is a fermionic Gaussian state: for a Gaussian
unitary $U$ and occupations $n_1,n_2\in\{0,1\}$,
\begin{equation}
  \frac{P_{++}}{\Tr P_{++}}
  =U\Bigl[\ketbra{n_1n_2}{n_1n_2}\otimes\frac{\id}{2^{m-2}}\Bigr]U^\dagger.
  \label{eq:sm-code-canonical}
\end{equation}

Let $V_{\rm iso}:\mathbb C^{2^{m-2}}\to\cH_m$ be an isometry with
$P_{++}=V_{\rm iso}V_{\rm iso}^\dagger$, and define the conditional blocks
\begin{equation}
  V_s=(\langle s|\otimes\id)V_{\rm iso},
  \qquad
  \tau_s=V_sV_s^\dagger
  =\langle s|P_{++}|s\rangle,
  \label{eq:sm-tau}
\end{equation}
so that $\Ran V_s=\supp\tau_s$ and
$\rank\tau_s\le\rank P_{++}=2^{m-2}$. Occupation postselection preserves
Gaussianity~\cite{Bravyi2005}, so every nonzero normalized block
$\widehat\tau_s=\tau_s/\Tr(\tau_s)$ is a Gaussian state on $m-1$ modes, with
product normal form
\begin{equation}
  \widehat\tau_s
  =U_s\Bigl[2^{-(m-1)}\prod_{j=1}^{m-1}(\id+\nu_jz_j)\Bigr]U_s^\dagger,
  \qquad|\nu_j|\le1,
  \label{eq:sm-product-form}
\end{equation}
where the $z_j$ are commuting quadratic reflections. The factor associated
with $j$ has rank two if $|\nu_j|<1$ and rank one if $|\nu_j|=1$, so
$\rank\widehat\tau_s=2^{\#\{j:|\nu_j|<1\}}$. Since
$\rank\tau_s\le2^{m-2}<2^{m-1}$, at least one $|\nu_j|=1$, and
\begin{equation}
  h_s=\sgn(\nu_j)\,U_sz_jU_s^\dagger
  \label{eq:sm-hs}
\end{equation}
is a quadratic reflection equal to $+1$ on $\supp\tau_s$; if $\tau_s=0$,
we can choose $h_s$ arbitrarily.

\begin{lemma}[Rank-two measurement bound]
\label{lem:sm-ranktwo}
Let $J$ contain two complex planes and be supported on the measured mode
$A$ together with an unmeasured subsystem $B_{\mathrm{loc}}$. For every
local state $\sigma$,
\begin{equation}
  \Tr(\sigma Q_J)
  \le1+\sum_sp_s\,S_1\!\bigl(\Gam_{\sigma_s}|_{B_{\mathrm{loc}}}\bigr).
  \label{eq:sm-ranktwo}
\end{equation}
\end{lemma}

\begin{proof}
Set $R_s=(\id+h_s)/2$ and $R=\sum_s\ketbra{s}{s}\otimes R_s$. Every code
vector $|v\rangle=V_{\rm iso}|u\rangle$ has $s$-component
$V_s|u\rangle\in\Ran V_s=\supp\tau_s\subseteq\Ran R_s$, hence
$\Ran P_{++}\subseteq\Ran R$, and since both are orthogonal projectors,
$P_{++}\le R$. Combining with Eq.~\eqref{eq:sm-QleP},
\begin{equation}
  Q_J\le2P_{++}\le2R=\id+\sum_s\ketbra{s}{s}\otimes h_s,
  \label{eq:sm-code-final}
\end{equation}
valid as an operator inequality, i.e., in every state. Taking the
expectation in $\sigma$, the $s$-th term equals
$p_s\Tr(\sigma_sh_s)\le p_sS_1(\Gam_{\sigma_s}|_{B_{\mathrm{loc}}})$ by
Lemma~\ref{lem:sm-kyfan}, each $h_s$ being a one-plane witness.
\end{proof}

We conclude with a few remarks. Gaussianity is required only of the
auxiliary code projector $P_{++}$, which is determined by the witness; the
physical state in which Eq.~\eqref{eq:sm-code-final} is evaluated may be
completely non-Gaussian. Conversely, the rank bound alone would not
suffice: a generic low-dimensional subspace need not be the eigenspace of
any Majorana bilinear, and it is the conjunction of Gaussianity and rank
deficiency that produces the branch reflections $h_s$. Gaussian closure
under occupation postselection is the only imported structural fact.

\section{Proof of Theorem~\ref{thm:main}}
\label{sm:proof-thm1}

\subsection{The one-mode partial sum inequality}
\label{sm:onemode}

\begin{lemma}[One-mode measurement]
\label{lem:sm-onemode}
Let $\rho$ be an arbitrary $N$-mode state and let $\{(p_s,\rho_s)\}_{s=0,1}$
be the ensemble obtained by measuring one occupation number,
Eq.~\eqref{eq:sm-conditional}. Then, for every $\ell=1,\ldots,N$,
\begin{equation}
  S_\ell(\Gam_\rho)\le1+\sum_{s}p_s\,S_{\ell-1}(\Gam'_{\rho_s}),
  \label{eq:sm-partialsum}
\end{equation}
with $\Gam'_{\rho_s}$ the branch covariance of the unmeasured modes.
\end{lemma}

\begin{proof}
Choose $J\in\cJ_\ell(V)$ maximizing Eq.~\eqref{eq:sm-kyfan} for
$\Gam_\rho$ and decompose it as in Lemma~\ref{lem:sm-localization}. Define
the complementary branch pairings $c_s=\frac12\tr(J_c^T\Gam'_{\rho_s})$.
Because $J_c$ is supported entirely in $B$, the affine identity
\eqref{eq:sm-affine} splits the maximal partial sum as
\begin{equation}
  S_\ell(\Gam_\rho)
  =\Tr(\sigma Q_{J_{\mathrm{loc}}})+\sum_sp_sc_s,
  \label{eq:sm-split}
\end{equation}
with $\sigma$ the reduced state of Eq.~\eqref{eq:sm-reduced}. Three cases
exhaust the local rank $r$.

If $r=0$, the first term vanishes and, since every Williamson value is at
most one, $c_s\le S_\ell(\Gam'_{\rho_s})\le1+S_{\ell-1}(\Gam'_{\rho_s})$;
averaging proves Eq.~\eqref{eq:sm-partialsum}.

If $r=1$, the local observable is a single Hermitian reflection, so
$\Tr(\sigma Q_{J_{\mathrm{loc}}})\le1$, while $J_c$ already contains
$\ell-1$ planes, so $c_s\le S_{\ell-1}(\Gam'_{\rho_s})$ by
Lemma~\ref{lem:sm-kyfan}.

If $r=2$, choose for each branch a one-plane structure $H_s$ on
$B_{\mathrm{loc}}$ attaining $S_1(\Gam_{\sigma_s})$. The direct sum
$J_c\oplus H_s$ is an $(\ell-1)$-plane partial complex structure on $B$,
whence
\begin{equation}
  c_s+S_1(\Gam_{\sigma_s})\le S_{\ell-1}(\Gam'_{\rho_s}),
  \label{eq:sm-completion}
\end{equation}
while Lemma~\ref{lem:sm-ranktwo}, evaluated on $\sigma$---whose branch
covariances are precisely the $\Gam_{\sigma_s}$---gives
$\Tr(\sigma Q_{J_{\mathrm{loc}}})\le1+\sum_sp_sS_1(\Gam_{\sigma_s})$.
Substituting the latter into Eq.~\eqref{eq:sm-split} and using
Eq.~\eqref{eq:sm-completion} proves Eq.~\eqref{eq:sm-partialsum} in the last
and hardest case.
\end{proof}

Boundary situations cause no difficulty: if $p_s=0$, the branch is
omitted; at $\ell=1$ the lemma reads $S_1(\Gam_\rho)\le1$, which is
covariance physicality, Eq.~\eqref{eq:sm-physicality}; and rank-deficient
conditional code blocks are harmless, cf.\ below
Eq.~\eqref{eq:sm-product-form}. We also note why a \emph{one}-mode
measurement is the natural elementary step: the finite-dimensional
reduction is tied to $\dim A=2$, which forces $r\le2$; a multimode
measurement is implemented sequentially, one mode at a time, with the
lemma applied conditionally after every outcome.

Since the recorded occupation state of the measured mode has definite
parity and a unit Williamson value, the branch spectrum with the record
retained is $\{1\}\sqcup\rvec(\Gam'_{\rho_s})$, and
Eq.~\eqref{eq:sm-partialsum} is equivalent to
\begin{equation}
  S_\ell(\Gam_\rho)\le\sum_sp_s\,S_\ell(\Gam_{\rho_s}),
  \label{eq:sm-partialsum-retained}
\end{equation}
which is the normalization used in Theorem~\ref{thm:main} of the main
text; Eq.~\eqref{eq:sm-partialsum} itself is Eq.~\eqref{eq:one-mode} there.

\subsection{From partial sum inequalities to weak majorization}
\label{sm:convex}

Let $x_j=r_j(\Gam_\rho)$ and $y_j=\sum_sp_sr_j(\Gam_{\rho_s})$, all lists
decreasing. The vector $y$ is decreasing because each branch spectrum is,
and its partial sums are $\sum_{j\le\ell}y_j=\sum_sp_sS_\ell(\Gam_{\rho_s})$,
so Eq.~\eqref{eq:sm-partialsum-retained} is exactly the statement
$x\prec_wy$, i.e., the vector form \eqref{eq:vector-majorization} of the
main text.

\begin{lemma}\label{lem:sm-transfer}
If $x\prec_wy$ for decreasing vectors with entries in $[0,1]$, then
$\sum_jf(x_j)\le\sum_jf(y_j)$ for every nondecreasing convex
$f:[0,1]\to\mathbb R$.
\end{lemma}

\begin{proof}
For a decreasing vector $x$, define $H_x(t)=\sum_j(x_j-t)_+$, where $(x)_+:=\max(x,0)$.
Largest partial sum dominance $x\prec_wy$ is equivalent to $H_x(t)\le H_y(t)$
for every $t\ge0$~\cite{Bhatia1997,MarshallOlkinArnold2011}. Any continuous convex
nondecreasing $f$ on $[0,1]$ admits the representation
\begin{equation}
  f(u)=f(0)+au+\int_0^1(u-t)_+\,{\rm d}\mu(t),
  \qquad a\ge0,\ \mu\ge0.
  \label{eq:sm-representation}
\end{equation}
Summing over components and using $H_x\le H_y$, together with
$\sum_jx_j\le\sum_jy_j$ for the linear term, proves the claim for continuous $f$. A nondecreasing convex $f$ on $[0,1]$ is continuous on $[0,1)$ and can be discontinuous only at the right endpoint: write $f=f_c+c\,\mathbf 1_{\{u=1\}}$ with $c=f(1)-f(1^-)\ge0$ and $f_c$ continuous, nondecreasing, and convex. The continuous part transfers as above. For the indicator, let $k=\#\{j:x_j=1\}$; then $\sum_{j\le k}x_j=k$, so $x\prec_wy$ forces $\sum_{j\le k}y_j\ge k$, and since every entry is at most one, $y_1=\cdots=y_k=1$. Hence $\#\{j:y_j=1\}\ge k$, which is the transfer inequality for $\mathbf 1_{\{u=1\}}$, completing the proof for arbitrary nondecreasing convex $f$.
\end{proof}

Applying scalar Jensen to each component,
$f(y_j)\le\sum_sp_sf(r_j(\Gam_{\rho_s}))$, and subtracting from $Nf(1)$
yields, for every nondecreasing convex $f$,
\begin{equation}
  \Phi_f(\rho)\ge\sum_sp_s\,\Phi_f(\rho_s)
  \label{eq:sm-deficit-monotone}
\end{equation}
under a one-mode occupation measurement with the record retained. This is
the workhorse used in Theorem~\ref{thm:admissible}(i).

\subsection{Lift to arbitrary Gaussian protocols}
\label{sm:lift}

Let us comment on how the result generalizes to arbitrary Gaussian operations, dispatched in the following one by one. (i)~Gaussian
unitaries act as $\Gam\mapsto O\Gam O^T$ with $O\in SO(2N)$ and preserve
every Williamson value. (ii)~Appending a pure Gaussian ancilla appends
unit Williamson values by Eq.~\eqref{eq:sm-concatenation}; unit entries
only increase partial sums and leave every spectral deficit unchanged.
(iii)~Occupation measurements obey Lemma~\ref{lem:sm-onemode}; adaptive
multimode measurements are sequences of one-mode measurements, with the
lemma applied conditionally at each node of the protocol tree, and the
nested conditional averages compose into the final ensemble average;
classical feedforward merely selects the branch and does not alter the
induction. An induction over the protocol tree therefore proves
the majorization~\eqref{eq:partialsum-main}--\eqref{eq:vector-majorization} for every record-retaining protocol without unrecorded discards, for arbitrary,
possibly mixed, inputs. Theorem~\ref{thm:main}, and its extension to protocols containing unrecorded discards, are derived in Sec.~\ref{sm:mixed}.

\section{Mixed states and recorded discards}
\label{sm:mixed}

We call a Gaussian protocol \emph{record-retaining} when its complete classical record is kept: measurement outcomes are never erased, and every discarded mode is first measured in its occupation basis with the outcome retained (a \emph{recorded} discard). An \emph{unrecorded} discard is a partial trace without measurement. Theorem~\ref{thm:main} of the main text is the pure-state case of the following general statement.

\begin{lemma}[Mixed-state Williamson majorization]
\label{lem:sm-mixed}
Let $\rho$ be an arbitrary $N$-mode state and let $\{(p_s,\rho_s)\}_s$ be the ensemble produced by any record-retaining Gaussian protocol. Then Eqs.~\eqref{eq:partialsum-main} and~\eqref{eq:vector-majorization} hold verbatim, with $\rho$ and $\rho_s$ in place of $\psi$ and $\phi_s$.
\end{lemma}

\begin{proof}
This is precisely the statement established by the induction of Sec.~\ref{sm:lift}: the one-mode Lemma~\ref{lem:sm-onemode} holds for arbitrary mixed states, and every elementary record-retaining operation preserves the partial sum inequalities.
\end{proof}

The restriction to record-retaining protocols is clear from the following example. Tracing out one mode of the two-mode pure Gaussian state $(|00\rangle+|11\rangle)/\sqrt2$ leaves a maximally mixed mode with $r=0$, so $S_1$ drops from $1$ to $0$: an unrecorded discard can create mixed Gaussian marginals on which covariance-based spectral deficits are no longer faithful, cf.~Sec.~\ref{sm:conventions}. For mixed output ensembles the classical record is therefore essential.

For pure state outputs, by contrast, unrecorded discards cause no loss of generality, and Theorem~\ref{thm:main} holds for arbitrary Gaussian protocols, as stated. Refine each unrecorded discard by measuring the discarded modes in the occupation basis while retaining the outcomes; the refined protocol is record-retaining, so Lemma~\ref{lem:sm-mixed} applies to the refined ensemble, and for a pure input every refined branch is pure. If a coarse-grained output branch is itself pure, all nonzero refined states contributing to it must coincide up to a phase---otherwise their mixture would have rank larger than one---so forgetting the refined record changes neither the output state nor the average of any spectral functional. The same collinearity argument covers classical coarse-graining of the record with pure outputs. Hence Theorem~\ref{thm:main}, and with it every pure-output statement of the main text, holds for arbitrary Gaussian protocols, including unrecorded discards.

\section{Proof of Theorem~\ref{thm:admissible}}
\label{sm:thm2}

\subsection{(i) Strong monotonicity if and only if \texorpdfstring{$f$}{f} is convex}

\emph{Sufficiency.} For convex nondecreasing $f$,
Eq.~\eqref{eq:sm-deficit-monotone} gives the elementary one-measurement
inequality, and the protocol lift of Sec.~\ref{sm:lift} extends it to
arbitrary Gaussian protocols with pure outputs, with unrecorded discards handled by the refinement of Sec.~\ref{sm:mixed}: $\Phi_f$ is a strong
monotone.

\emph{Necessity.} We exhibit, for every failure of convexity, a Gaussian
protocol that increases $\Phi_f$ on average. Fix $x,y\in[0,1]$ and
$p\in(0,1)$, $q=1-p$. Choose real single-qubit states
$v_s=\cos\theta_s|0\rangle+\sin\theta_s|1\rangle$ with
\begin{equation}
  \langle v_0|Z|v_0\rangle=\cos2\theta_0=x,
  \qquad
  \langle v_1|Z|v_1\rangle=\cos2\theta_1=-y,
  \label{eq:sm-angles}
\end{equation}
and consider the two-mode pure state
\begin{equation}
  |\psi_{x,y,p}\rangle
  =\sqrt{p}\,|0\rangle|v_0\rangle+\sqrt{q}\,|1\rangle|v_1\rangle.
  \label{eq:sm-family}
\end{equation}
With the Jordan--Wigner Majoranas $\gamma_1=X\otimes\id$,
$\gamma_2=Y\otimes\id$, $\gamma_3=Z\otimes X$, $\gamma_4=Z\otimes Y$, and
writing $a=x$, $b=-y$, a direct evaluation of the covariance gives
\begin{equation}
  \Gam_\psi=
  \begin{pmatrix}
  0&p-q&0&2\sqrt{pq}\,\sin(\theta_1{-}\theta_0)\\
  q-p&0&2\sqrt{pq}\,\sin(\theta_0{+}\theta_1)&0\\
  0&-2\sqrt{pq}\,\sin(\theta_0{+}\theta_1)&0&pa+qb\\
  -2\sqrt{pq}\,\sin(\theta_1{-}\theta_0)&0&-(pa+qb)&0
  \end{pmatrix}.
  \label{eq:sm-family-covariance}
\end{equation}
Using $\sin^2(\theta_1{-}\theta_0)+\sin^2(\theta_0{+}\theta_1)=1-ab$ and
$\sin(\theta_1{-}\theta_0)\sin(\theta_0{+}\theta_1)=(a-b)/2$, one finds
\begin{equation}
  r_1^2+r_2^2=(p-q)^2+(pa+qb)^2+4pq(1-ab)=1+(pa-qb)^2,
  \qquad
  r_1r_2=|{\rm Pf}(\Gam_\psi)|=|pa-qb|,
  \label{eq:sm-family-invariants}
\end{equation}
so the Williamson spectrum is exactly
\begin{equation}
  \rvec(\Gam_{\psi_{x,y,p}})=\bigl(1,\;px+qy\bigr).
  \label{eq:sm-family-spectrum}
\end{equation}
Measuring the occupation of the first mode yields outcome $s=0$ with
probability $p$ and $s=1$ with probability $q$, with pure post-measurement
states $|s\rangle|v_s\rangle$ whose spectra are $(1,x)$ and $(1,y)$,
respectively. Strong monotonicity of $\Phi_f$ applied to this protocol
therefore demands
\begin{equation}
  f(1)-f(px+qy)\ \ge\ p\bigl[f(1)-f(x)\bigr]+q\bigl[f(1)-f(y)\bigr],
  \qquad\text{i.e.}\qquad
  f(px+qy)\le pf(x)+qf(y).
  \label{eq:sm-convexity-forced}
\end{equation}
Ranging over all $x,y\in[0,1]$ and $p\in(0,1)$, this is precisely
convexity of $f$; any nonconvex nondecreasing $f$ is violated by some
member of the family. The witnesses~\eqref{eq:sm-family} carry parity coherence; within the parity-superselected sector only the sufficiency direction is used in this work, and it holds there verbatim.\hfill$\square$

\subsection{(ii) Additivity}

If at least one factor of $\psi\otimes\phi$ has definite parity, the
covariance is block diagonal, Eq.~\eqref{eq:sm-concatenation}, the
Williamson multisets concatenate, and
\begin{equation}
  \Phi_f(\psi\otimes\phi)=\Phi_f(\psi)+\Phi_f(\phi)
  \label{eq:sm-additive}
\end{equation}
for every $f$, since $\Phi_f$ is a sum over the spectrum. In the
parity-superselected setting this applies to all physical
states.\hfill$\square$

\subsection{(iii) Faithfulness}

Each term of $\Phi_f=\sum_j[f(1)-f(r_j)]$ is nonnegative for nondecreasing
$f$. If $f(r)<f(1)$ for every $r<1$, then $\Phi_f(\psi)=0$ forces
$r_j=1$ for all $j$, which by Eq.~\eqref{eq:sm-faithfulness} holds exactly
for pure Gaussian states. Conversely, if $f(r_*)=f(1)$ for some $r_*<1$,
the product state $|0\rangle|v_0\rangle$ of the family
\eqref{eq:sm-family} with $x=r_*$ has spectrum $(1,r_*)$: it is pure,
non-Gaussian, and satisfies $\Phi_f=0$, so $\Phi_f$ is not
faithful.\hfill$\square$

\subsection{(iv) Asymptotic continuity}

Define the uniform modulus of continuity
$\omega_f(\delta)=\sup\{|f(u)-f(v)|:u,v\in[0,1],|u-v|\le\delta\}$; by
compactness of $[0,1]$, continuity of $f$ implies
$\omega_f(\delta)\to0$ as $\delta\to0$. For states $\rho,\sigma$ on the
same $N$-mode system with $\epsilon=\|\rho-\sigma\|_1$, covariance
stability gives $\|\Gam_\rho-\Gam_\sigma\|_{\rm op}\le\epsilon$~\cite{bittel2025optimal},
and Weyl perturbation of singular values then yields
$|r_j(\rho)-r_j(\sigma)|\le\epsilon$ for every $j$. Hence
\begin{equation}
  |\Phi_f(\rho)-\Phi_f(\sigma)|
  \le\sum_{j=1}^N\bigl|f(r_j(\rho))-f(r_j(\sigma))\bigr|
  \le N\,\omega_f(\epsilon),
  \label{eq:sm-continuity}
\end{equation}
so that for sequences with $\epsilon_N\to0$ the deficit densities
converge, $|\Phi_f(\rho_N)-\Phi_f(\sigma_N)|/N\le\omega_f(\epsilon_N)\to0$:
$\Phi_f$ is asymptotically continuous~\cite{ChitambarGour2019}.

\section{Alternative proof of the monotonicity of the Gaussian nullity under post-selection}
\label{sm:nullity}

We give an alternative, self-contained proof that the Gaussian nullity
$\nuG(\psi)=\#\{j:r_j<1\}$ of a pure state cannot increase in \emph{any}
nonzero branch of a one-mode occupation measurement---a statement stronger
than strong monotonicity. The proof identifies Williamson values of unity
with linear Majorana annihilators.

\begin{lemma}[Linear annihilators]
\label{lem:sm-annihilators}
For $z\in\mathbb C^{2N}$ set $\gamma(z)=\sum_az_a\gamma_a$ and
$\mathcal A_\psi=\{z:\gamma(z)|\psi\rangle=0\}$. For every normalized pure
state,
\begin{equation}
  \mathcal A_\psi=\ker(\id+i\Gam_\psi),
  \qquad
  \dim_{\mathbb C}\mathcal A_\psi=\#\{j:r_j=1\},
  \qquad
  \nuG(\psi)=N-\dim_{\mathbb C}\mathcal A_\psi.
  \label{eq:sm-annihilator-count}
\end{equation}
\end{lemma}

\begin{proof}
From the covariance definition,
$\langle\psi|\gamma_a\gamma_b|\psi\rangle=\delta_{ab}+i(\Gam_\psi)_{ab}$,
so $\|\gamma(z)|\psi\rangle\|^2=z^\dagger(\id+i\Gam_\psi)z$ and
$\mathcal A_\psi=\ker(\id+i\Gam_\psi)$. A Williamson plane contributes the
block $\bigl(\begin{smallmatrix}1&ir_j\\-ir_j&1\end{smallmatrix}\bigr)$
with determinant $1-r_j^2$, which has a one-dimensional complex kernel
exactly when $r_j=1$.
\end{proof}

\begin{lemma}[Branch counting]
\label{lem:sm-branch}
Let $|\psi\rangle=\sqrt{p_0}|0\rangle|\phi_0\rangle
+\sqrt{p_1}|1\rangle|\phi_1\rangle$ be an arbitrary pure state, with no
parity restriction, and let $|\psi_s\rangle=|s\rangle|\phi_s\rangle$ be
the branches with $p_s>0$. Then $\nuG(\psi_s)\le\nuG(\psi)$.
\end{lemma}

\begin{proof}
Order the measured mode first and use Eq.~\eqref{eq:sm-strings}; with
$f=(\gamma_1+i\gamma_2)/2=|0\rangle\!\langle1|\otimes\id$, every linear
Majorana operator decomposes uniquely as $L=\alpha f+\beta f^\dagger
+Z\otimes W$, with $W$ linear in the intrinsic Majoranas of $B$. Acting on
$\psi$,
\begin{equation}
  L|\psi\rangle
  =|0\rangle\bigl(\alpha\sqrt{p_1}|\phi_1\rangle
  +\sqrt{p_0}W|\phi_0\rangle\bigr)
  +|1\rangle\bigl(\beta\sqrt{p_0}|\phi_0\rangle
  -\sqrt{p_1}W|\phi_1\rangle\bigr),
  \label{eq:sm-L-action}
\end{equation}
and $L\in\mathcal A_\psi$ forces both parentheses to vanish. Let
$d=\dim\mathcal A_\psi$ and suppose $p_0>0$. The coefficient map
$L\mapsto\alpha$ on $\mathcal A_\psi$ has image of dimension at most one,
so its kernel $\mathcal K_0=\{L\in\mathcal A_\psi:\alpha=0\}$ has
$\dim\mathcal K_0\ge d-1$ by rank--nullity. For $L\in\mathcal K_0$, the
first parenthesis of Eq.~\eqref{eq:sm-L-action} gives
$W|\phi_0\rangle=0$, i.e., $W\in\mathcal A_{\phi_0}$; the map
$\mathcal K_0\to\mathcal A_{\phi_0}$, $L\mapsto W$, is injective, because
$W=0$ leaves $L=\beta f^\dagger$ and the second parenthesis then forces
$\beta=0$. Hence $\dim\mathcal A_{\phi_0}\ge d-1$. For the product branch
$|\psi_0\rangle=|0\rangle|\phi_0\rangle$ one checks directly that
$L|\psi_0\rangle=\beta|1\rangle|\phi_0\rangle+|0\rangle W|\phi_0\rangle$,
so $\mathcal A_{\psi_0}=\operatorname{span}\{f\}\oplus
(Z\otimes\mathcal A_{\phi_0})$, where $Z\otimes\mathcal A_{\phi_0}$
denotes the operators $Z\otimes W$ with $W\in\mathcal A_{\phi_0}$, and
\begin{equation}
  \dim\mathcal A_{\psi_0}=1+\dim\mathcal A_{\phi_0}\ge d.
  \label{eq:sm-branch-dim}
\end{equation}
The case $s=1$ is identical with $\beta$ and $f^\dagger$ in place of
$\alpha$ and $f$. By Lemma~\ref{lem:sm-annihilators},
$\nuG(\psi_s)=N-\dim\mathcal A_{\psi_s}\le N-d=\nuG(\psi)$.
\end{proof}

Monotonicity under postselection implies average monotonicity,
$\sum_sp_s\nuG(\psi_s)\le\nuG(\psi)$; Gaussian unitaries preserve
$\nuG$, definite-parity ancillas append saturated modes and leave it
unchanged, and the protocol induction of Sec.~\ref{sm:lift}---including
the discard refinement of Sec.~\ref{sm:mixed} for pure outputs---then shows that $\nuG$ is a
strong monotone.

\section{Proof of Theorem~\ref{thm:catalytic-bound}}
\label{sm:catalysis}

Using the unit-padding convention, we may take $\psi$ and $\phi$ to have
the same number $n$ of modes; let the catalyst have $m$ modes. Denote the
real Majorana spaces of the system and catalyst by
$V_S\simeq\mathbb R^{2n}$ and $V_C\simeq\mathbb R^{2m}$. Let
$W=v_\omega^\perp$ when $v_\omega\ne0$, and otherwise let $W$ be any
codimension-one subspace of $V_C$. In either case, $\dim W=2m-1$.

Let $B_W\in\mathbb R^{2m\times(2m-1)}$ have orthonormal columns spanning
$W$, and define the isometry
\begin{equation}
  R=\begin{pmatrix}
      \id_{2n}&0\\
      0&B_W
    \end{pmatrix},
  \qquad R^TR=\id_{2(n+m)-1}.
  \label{eq:sm-catalyst-isometry}
\end{equation}
Since $B_W^Tv_\omega=0$, direct substitution into the product block form
in Eq.~\eqref{eq:product-cross-block} gives
\begin{align}
  C_\rho
  &:=R^T\Gam_{\rho\otimes\omega}R \notag\\
  &=\begin{pmatrix}
      \Gam_\rho & u_\rho v_\omega^TB_W\\
      -B_W^Tv_\omega u_\rho^T & B_W^T\Gam_\omega B_W
    \end{pmatrix} \notag\\
  &=\Gam_\rho\oplus\Gam_{\omega,W},
  \qquad
  \Gam_{\omega,W}:=B_W^T\Gam_\omega B_W.
  \label{eq:sm-catalyst-compression}
\end{align}
Thus, $iC_\rho$ is a codimension-one
Hermitian principal compression of $i\Gam_{\rho\otimes\omega}$.

Because $\Gam_{\omega,W}$ is an odd-dimensional real antisymmetric
matrix, its singular values consist of $m-1$ equal pairs and one
unpaired zero. Let $\bm d_\omega=(d_1,\ldots,d_{m-1})$ contain one value
from each pair, in decreasing order, and set
\begin{equation}
  z_\rho:=\widetilde \rvec(\Gam_\rho)\sqcup\bm d_\omega,
  \qquad N:=n+m,
  \label{eq:sm-catalyst-z}
\end{equation}
where the union is rearranged in decreasing order. The eigenvalues of
$i\Gam_{\rho\otimes\omega}$ and $iC_\rho$, respectively, are
\begin{align}
  &\widetilde r_1,\ldots,\widetilde r_N,-\widetilde r_N,\ldots,-\widetilde r_1, \notag\\
  &(z_\rho)_1,\ldots,(z_\rho)_{N-1},0,
    -(z_\rho)_{N-1},\ldots,-(z_\rho)_1,
  \label{eq:sm-catalyst-spectra}
\end{align}
with $\widetilde r_j= \widetilde r_j(\Gam_{\rho\otimes\omega})$. Cauchy interlacing theorem~\cite{Bhatia1997} for the
Hermitian principal compression of $i\Gam_{\rho\otimes\omega}$ then yields
\begin{equation}
  \widetilde r_j(\Gam_{\rho\otimes\omega})
  \ge (z_\rho)_j
  \ge \widetilde r_{j+1}(\Gam_{\rho\otimes\omega}),
  \qquad j=1,\ldots,N-1.
  \label{eq:sm-catalyst-interlacing}
\end{equation}

Write
$p=\widetilde \rvec(\Gam_{\psi\otimes\omega})$ and
$q=\widetilde \rvec(\Gam_{\phi\otimes\omega})$. The necessary condition for the
catalyzed deterministic conversion is $p\prec_w q$. From the first half
of Eq.~\eqref{eq:sm-catalyst-interlacing} for the input and the second
half for the output,
\begin{equation}
  (z_\psi)_j\le p_j,
  \qquad
  q_j\le(z_\phi)_{j-1}\quad(j=2,\ldots,N).
  \label{eq:sm-catalyst-two-sides}
\end{equation}
Consequently, for $k=1,\ldots,N-1$,
\begin{equation}
  S_k(z_\psi)
  \le S_k(p)
  \le S_k(q)
  \le 1+S_{k-1}(z_\phi),
  \label{eq:sm-catalyst-partialsum}
\end{equation}
where the last inequality uses $q_1\le1$ and
$S_0:=0$. These inequalities, together with the trivial final partial sum 
inequality obtained from the case $k=N-1$, are precisely
\begin{equation}
  z_\psi\sqcup(0)\prec_w(1)\sqcup z_\phi.
  \label{eq:sm-catalyst-common}
\end{equation}
Substituting Eq.~\eqref{eq:sm-catalyst-z} gives
\begin{equation}
  \widetilde \rvec(\Gam_\psi)\sqcup\bm d_\omega\sqcup(0)
  \prec_w
  (1)\sqcup\widetilde \rvec(\Gam_\phi)\sqcup\bm d_\omega.
  \label{eq:sm-catalyst-before-cancel}
\end{equation}
Finally, $x\prec_w y$ is equivalent to
$\sum_j f(x_j)\le\sum_j f(y_j)$ for every convex nondecreasing
$f$~\cite{MarshallOlkinArnold2011}. The identical term
$\sum_a f(d_a)$ cancels from the two sides of
Eq.~\eqref{eq:sm-catalyst-before-cancel}, leaving
\begin{equation}
  \widetilde \rvec(\Gam_\psi)\sqcup(0)
  \prec_w
  (1)\sqcup\widetilde \rvec(\Gam_\phi),
  \label{eq:sm-catalyst-final}
\end{equation}
which is the claim.
\hfill$\square$

\section{Asymptotic conversion bounds and irreversibility}
\label{sm:asymptotic}

Throughout this section the input and target states are assumed to have
definite fermionic parity, as demanded by parity superselection; this is
the only place where a parity assumption enters, through the additivity
statement of Theorem~\ref{thm:admissible}(ii), which guarantees 
$\Phi_f(\psi^{\otimes n})=n\Phi_f(\psi)$.

\subsection{Proof of Corollary~\ref{cor:rate}}

Call $f$ admissible as in Theorem~\ref{thm:admissible}: continuous,
nondecreasing, convex on $[0,1]$, with $f(r)<f(1)$ for $r<1$.

\begin{lemma}[Rate bound]
\label{lem:sm-rate}
For every admissible $f$ with $\Phi_f(\phi)>0$,
\begin{equation}
  R(\psi\to\phi)\le\frac{\Phi_f(\psi)}{\Phi_f(\phi)} .
  \label{eq:sm-rate}
\end{equation}
\end{lemma}

\begin{proof}
Let the protocol map $\psi^{\otimes n}$ to a pure output
$\widetilde\phi_n$ with
$\|\widetilde\phi_n-\phi^{\otimes m_n}\|_1\le\epsilon_n\to0$, as in
Eq.~\eqref{eq:rate-def} of the main text. Strong monotonicity and
additivity give
\begin{equation}
  n\,\Phi_f(\psi)=\Phi_f(\psi^{\otimes n})\ge\Phi_f(\widetilde\phi_n).
  \label{eq:sm-rate-step1}
\end{equation}
The ideal target lives on $m_nN_\phi$ modes, with $N_\phi$ the mode number
of $\phi$, so asymptotic continuity, Eq.~\eqref{eq:sm-continuity}, yields
\begin{equation}
  \Phi_f(\widetilde\phi_n)
  \ge m_n\Phi_f(\phi)-m_nN_\phi\,\omega_f(\epsilon_n).
  \label{eq:sm-rate-step2}
\end{equation}
Combining, for all $n$ large enough that
$N_\phi\omega_f(\epsilon_n)<\Phi_f(\phi)$---which $\epsilon_n\to0$
guarantees eventually---one has
$m_n/n\le\Phi_f(\psi)/[\Phi_f(\phi)-N_\phi\omega_f(\epsilon_n)]$; since
only the limit superior matters, these $n$ suffice, and
$\omega_f(\epsilon_n)\to0$ proves Eq.~\eqref{eq:sm-rate}.
\end{proof}

Optimizing over all admissible $f$ proves
Corollary~\ref{cor:rate}. 

\subsection{Finite optimization of the asymptotic bound}
\label{sm:finite-opt}

Throughout this subsection $f_c(r):=\max\{r,c\}$ for $c\in[0,1)$, and
$W_c:=\Phi_{f_c}$, so that
$W_c(\rho)=\sum_j\min\{1-r_j(\Gam_\rho),1-c\}$. Each $f_c$ is continuous,
nondecreasing and convex, and $f_c(r)=\max\{r,c\}<1=f_c(1)$ for $r<1$;
hence every $f_c$ is admissible.

By standard results on convex functions and
Lebesgue--Stieltjes measures~\cite{RockafellarConvex}, every admissible
$f$ admits the representation
\begin{equation}
  f(r)=f(0)+b\,r+\int_{(0,1)}(r-c)_+\,d\mu(c),
  \qquad r\in[0,1],
  \label{eq:sm-hinge}
\end{equation}
where $b\geq0$ and $\mu$ is
a locally finite nonnegative Borel measure on $(0,1)$ satisfying
$\int_{(0,1)}(1-c)\,d\mu(c)<\infty$. Substituting Eq.~\eqref{eq:sm-hinge} into $f(1)-f(r)$ gives
\begin{equation}
f(1)-f(r)
=
b(1-r)
+
\int_{(0,1)}
\bigl[(1-c)-(r-c)_+\bigr]\,d\mu(c).
\end{equation}
Since
$(1-c)-(r-c)_+=\min\{1-r,1-c\}$, summing over the Williamson values of
$\rho$ yields
\begin{equation}
\Phi_f(\rho)
=
bW_0(\rho)+\int_{(0,1)}W_c(\rho)\,d\mu(c)
=
\int_{[0,1)}W_c(\rho)\,d\nu(c),
\label{eq:sm-family-decomp}
\end{equation}
where $\nu=b\,\delta_0+\mu$.

\begin{lemma}[Reduction to the threshold set]
\label{lem:threshold}
Let $\phi$ be non-Gaussian. Then $c\mapsto W_c(\psi)/W_c(\phi)$ attains its
infimum over $[0,1)$ at a point of
$\mathcal C_{\psi,\phi}
=\{0\}
\cup\{r_j(\Gam_\psi):r_j(\Gam_\psi)<1\}
\cup\{r_j(\Gam_\phi):r_j(\Gam_\phi)<1\}$,
a set of at most $2N+1$ points.
\end{lemma}

\begin{proof}
List the distinct elements of $\mathcal C_{\psi,\phi}$ as
$0=c_0<c_1<\cdots<c_m<1$ and set $c_{m+1}=1$. Write
$n_\rho(c)=\#\{j:r_j(\Gam_\rho)\leq c\}$ and
$A_\rho(c)=\sum_{r_j(\Gam_\rho)>c}[1-r_j(\Gam_\rho)]$, so that
$W_c(\rho)=A_\rho(c)+(1-c)n_\rho(c)$.

 On each interval $[c_i,c_{i+1})$, the quantities $A_\rho$ and $n_\rho$
are constant. Hence
\begin{equation}
g(c)
=
\frac{A_\psi+(1-c)n_\psi}
     {A_\phi+(1-c)n_\phi},
\qquad
g'(c)
=
\frac{n_\phi A_\psi-n_\psi A_\phi}
     {[A_\phi+(1-c)n_\phi]^2}.
\end{equation}
Since $\phi$ is non-Gaussian, $W_c(\phi)>0$ for every $c<1$, so the
denominator does not vanish. The derivative therefore has constant sign.
Although $A_\rho$ and $n_\rho$ may change at $c_{i+1}$, continuity of
$W_c(\rho)$ ensures that the limiting value of the ratio agrees with its
value at $c_{i+1}$. Its minimum on this
interval is thus attained at one of the endpoints.

It remains to consider the final interval $[c_m,1)$. For every
$c\in[c_m,1)$, all nonunit Williamson values of either state satisfy
$r_j\leq c$, so that $A_\rho(c)=0$ and
$n_\rho(c)=\nuG(\rho)$. Therefore
$W_c(\rho)=(1-c)\nuG(\rho)$ and
\begin{equation}
  \frac{W_c(\psi)}{W_c(\phi)}=\frac{\nuG(\psi)}{\nuG(\phi)}
  \qquad\text{for all }c\in[c_m,1).
  \label{eq:sm-constant-tail}
\end{equation}
The ratio is thus constant on the final interval and its value is already
attained at $c_m\in\mathcal C_{\psi,\phi}$. The infimum over $[0,1)$ is
therefore a minimum over $\mathcal C_{\psi,\phi}$, which contains at most
the point $0$ and at most $N$ nonunit Williamson values from each state.
\end{proof}

\begin{proof}[Proof of Proposition~\ref{prop:finite_opt}]
Let $f$ be admissible with $\Phi_f(\phi)>0$, and set
$\nu=b\,\delta_0+\mu$. By Eq.~\eqref{eq:sm-family-decomp}, for
$\rho\in\{\psi,\phi\}$ one has
$\Phi_f(\rho)=\int_{[0,1)}W_c(\rho)\,d\nu(c)$.

Since $\phi$ is non-Gaussian, one has
$W_c(\phi)>0$ for every $c<1$. We may therefore define the measure
$dP_f(c):=W_c(\phi)\,d\nu(c)/\Phi_f(\phi)$. Its total mass is
$\int_{[0,1)}dP_f(c)
=\int_{[0,1)}W_c(\phi)\,d\nu(c)/\Phi_f(\phi)=1$, so $P_f$ is a
probability measure.

Using
$W_c(\psi)
=[W_c(\psi)/W_c(\phi)]W_c(\phi)$, we obtain
$\Phi_f(\psi)/\Phi_f(\phi)
=\int_{[0,1)}[W_c(\psi)/W_c(\phi)]\,dP_f(c)$. Thus the ratio
$\Phi_f(\psi)/\Phi_f(\phi)$ is a probability-weighted average of the
pointwise ratios $W_c(\psi)/W_c(\phi)$. Such an average cannot be smaller
than the infimum of the quantities being averaged. Hence
\begin{equation}
  \frac{\Phi_f(\psi)}{\Phi_f(\phi)}
  \ \ge\ \inf_{c\in[0,1)}\frac{W_c(\psi)}{W_c(\phi)}
  \ =\ \min_{c\in\mathcal C_{\psi,\phi}}\frac{W_c(\psi)}{W_c(\phi)},
\end{equation}
the last equality following from Lemma~\ref{lem:threshold}. This gives
Eq.~\eqref{eq:finite-Williamson-bound}.
\end{proof}

\subsection{Irreversibility: proof of Theorem~\ref{thm:irrev}}

Let $\omega$ be a pure state on $N$ modes with $\Gam_\omega=0$ and $\chi$
any non-Gaussian pure state. In the main text, we have established
\begin{equation}
  R(\omega\to\chi)\,R(\chi\to\omega)
  \le\frac{W_0(\chi)}{\nuG(\chi)}
  \le1.
  \label{eq:sm-roundtrip}
\end{equation}
 The last
inequality is strict whenever $\chi$ carries a Williamson value in the
open interval $(0,1)$, which rules out reversible interconversion for
every such pair. Moreover, the ratio can be made arbitrarily small within
the definite-parity sector: the four-mode even-parity states
\begin{equation}
  |\chi_t\rangle=\cos t\,|0000\rangle+\sin t\,|1111\rangle,
  \qquad 0<t<\frac{\pi}{4},
  \label{eq:sm-chi-family}
\end{equation}
have all four Williamson values equal to $\cos2t<1$, so they are
non-Gaussian by Eq.~\eqref{eq:sm-faithfulness}, with $\nuG(\chi_t)=4$ and
$W_0(\chi_t)=4(1-\cos2t)$, whence
$W_0(\chi_t)/\nuG(\chi_t)=1-\cos2t\to0$ as $t\to0$. The product of the
two conversion rates therefore admits no positive lower bound, proving
Theorem~\ref{thm:irrev}.\hfill$\square$

\end{document}